\documentclass[11pt]{amsart}
\usepackage{graphicx} % Requiblue for inserting images

\usepackage{fullpage}
\usepackage{xcolor}
\usepackage{makecell}
\usepackage{soul}
\usepackage{wrapfig}
\newcommand{\erclogowrapped}[1]{%
\setlength\intextsep{0pt}%
\begin{wrapfigure}[3]{r}{#1}%
\includegraphics[width=#1]{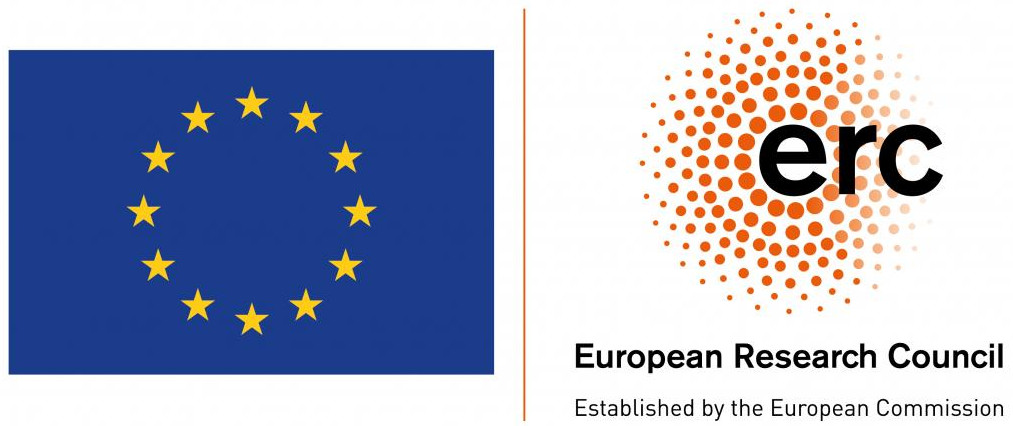}%
\end{wrapfigure}%
}

\newcommand{\paren}[1]{\left( #1 \right)}
\newcommand{\counting}{M_{\mathsf{count}}}
\newcommand{\norm}[1]{\left\| #1 \right\|}

\newcommand{\set}[1]{\left\{ #1 \right\}}
\newcommand{\op}[1]{\operatorname{#1}}
\newcommand{\col}{{1 \to 2}}
\newcommand{\row}{{2 \to \infty}}

\newcommand{\meanSE}{\text{err}_{2}}
\newcommand{\maxSE}{\text{err}_{\infty}}

\renewcommand{\meanSE}{\text{MeanSE}}
\renewcommand{\maxSE}{\text{MaxSE}}

\newcommand{\MeanSEUpperBound}{\frac{\log (n)}{\pi} + \underbrace{\alpha_{\infty} - \frac{1}{\pi}}_{\approx 0.748} + o(1)}
\newcommand{\MeanSEUpperBoundExplicit}{{0.748} + \frac{\log (n)}{\pi} +  o(1)}

\newcommand{\MaxSEUpperBoundExplicit}{0.846 + \frac{\log (n)}{\pi} + o(1)}

\usepackage[colorlinks=true, allcolors=blue]{hyperref}
\usepackage{amssymb}
\usepackage{bbm}
\usepackage{amsthm, thm-restate}
\newtheorem{theorem}{Theorem}

\newtheorem{lemma}{Lemma}

\newtheorem{corollary}{Corollary}

\newtheorem{definition}{Definition}
\usepackage{graphicx}
\usepackage{subcaption}
\usepackage{booktabs}
\usepackage{comment}

\title{Normalized Square Root: Sharper Matrix Factorization Bounds for Differentially Private Continual Counting}
\author{Monika Henzinger}
\thanks{Institute of Science and Technology (ISTA), Austria. \texttt{monika.henzinger@ist.ac.at}}

\author{Nikita P. Kalinin}
\thanks{Institute of Science and Technology (ISTA), Austria. \texttt{nikita.kalinin@ist.ac.at}}

\author{Jalaj Upadhyay}
\thanks{Rutgers University, USA \texttt{jalaj.upadhyay@rutgers.edu}}
\begin{document}

\begin{abstract}
   The factorization norms of the lower-triangular all-ones $n\times n$ matrix, $\gamma_2(M_{\mathsf{count}})$ and $\gamma_{\mathrm{F}}(M_{\mathsf{count}})$, play a central role in differential privacy as they are used to give theoretical justification of the accuracy of the only known production-level private training algorithm of deep neural networks by Google. Prior to this work, the best known upper bound on $\gamma_2(M_{\mathsf{count}})$ was $1 + \tfrac{\log(n)}{\pi}$ by Mathias (Linear Algebra and Applications, 1993), and the best known lower bound was $\tfrac{1}{\pi}\bigl(2 + \log\bigl(\tfrac{2n+1}{3}\bigr)\bigr) \approx 0.507 + \tfrac{\log(n)}{\pi}$ (Matou{\v{s}}ek, Nikolov, Talwar, IMRN 2020), where $\log(\cdot)$ is the natural logarithm. Recently, Henzinger and Upadhyay (SODA 2025) gave the first explicit factorization that meets the bound of Mathias (1993) and asked whether there exists an explicit factorization that improves on Mathias’ bound. We answer this question in the affirmative. Additionally, we improve the lower bound significantly. More specifically, we show that
   $$
      o(1) + 0.701 + \frac{\log(n)}{\pi} \;\leq\; \gamma_2(M_{\mathsf{count}}) \;\leq\; 0.846 + \frac{\log(n)}{\pi} + o(1).
   $$
   That is, we reduce the gap between the upper and lower bound to $0.14 + o(1)$.

   We show that our factors achieve a better upper bound for $\gamma_{\mathrm{F}}(M_{\mathsf{count}})$ compared to prior work, and we also establish an improved lower bound for $\gamma_{\mathrm{F}}(M_{\mathsf{count}})$:
   $$
      o(1) + 0.701 + \frac{\log(n)}{\pi} \;\leq\; \gamma_{\mathrm{F}}(M_{\mathsf{count}}) \;\leq\; 0.748 + \frac{\log(n)}{\pi} + o(1).
   $$
   That is, the gap between the lower and upper bound provided by our explicit factorization is $0.047 + o(1)$.
\end{abstract}

\maketitle
%\thispagestyle{empty}

%\clearpage
\setcounter{page}{1}
\section{Introduction}
\label{sec:introduction}
The \textit{matrix mechanism}~\cite{li2015matrix} is a foundational tool in differential privacy that enables the accurate release of a set of linear queries (called the {\em workload matrix}) while minimizing added noise. Given a workload matrix $A$, the matrix mechanism, $\mathcal M_{\mathsf{mm}}$, first computes a factor $A=LR$ and then on input $x$ outputs $\mathcal M_{\mathsf{mm}}(A;x)=L(Rx+z)$, where $z$ is an appropriately scaled multivariate Gaussian vector.  With a suitable scaling, the mechanism satisfies $\mu$-Gaussian differential privacy (see Definition~\ref{defn:DP}).  In the recent years, a special workload matrix, namely the lower-triangular all-ones matrix, i.e.,
\begin{align}
\label{eq:countingMatrix}
\counting[i,j] = \begin{cases}
    1 & i \geq j \\
    0 & i < j
\end{cases}    
\end{align}
has gained special interest due to its application in the private training of neural networks~\cite{mcmahan2022private}. The matrix mechanism improves model accuracy under the strict privacy constraints of differential privacy by applying a carefully chosen linear transformation to the gradients and then privatizing it (see \cite{pillutla2025correlated}).

To evaluate the accuracy of the matrix mechanism in privately training neural networks, two widely used {metrics} are the \textit{Mean Squared Error} (MeanSE) and \textit{Maximum Squared Error} (MaxSE). For a randomized algorithm $\mathcal M$ and $\counting\in \{0,1\}^{n\times n}$ as defined above, they are defined as follows:
\begin{align}
    \begin{split}
        \maxSE(\mathcal M;n) &:= \max_i \paren{\mathbb E_{\mathcal M}[(\mathcal M(x)-\counting x)^2_i]}^{1/2} \\
        \meanSE(\mathcal M;n) &:=  \mathbb E_{\mathcal M} \left[\tfrac{1}{n}\norm{\mathcal M(x) - \counting x}_2^2\right]^{1/2}
    \end{split}
\end{align}

While $\meanSE(\mathcal M;n)$ reflects the {average accuracy} of the mechanism, $\maxSE(\mathcal M;n)$ helps identify potential bottlenecks in convergence or stability.  Together, these metrics provide complementary insights into how the matrix mechanism balances privacy and accuracy during differentially private model training. For example, training neural networks privately is an iterative process that takes a total so-called {\em privacy budget}, and a scaled Gaussian noise (known as {\em scale multiplier}) is added to the estimated gradients in every iteration. This process is stopped once the total privacy budget is consumed. Since the stopping time is not known {\it a prori}, having a high MaxSE runs into the risk of {outputting a model with large accuracy loss}. Therefore, to get a proper understanding of private training, ideally, one would like a mechanism to give the best bounds for both these error metrics.

%- Both are measublue in terms of associated factorization norms
The $\meanSE$ and $\maxSE$ error of the matrix mechanism $\mathcal M_{\mathsf{mm}}$ can be upper bounded in terms of the appropriate factorization norms, denoted by $\gamma_{\op F}(\cdot)$ and $\gamma_{\op 2}(\cdot)$-norm of matrix $\counting$~\cite{edmonds2020power}. Let $\mu$ be the scale multiplier of the zero-mean Gaussian distribution required to ensure $(1/\mu)$-Gaussian differential privacy (see Definition~\ref{def:gaussian_mech}). For a matrix $A$, let $\norm{A}_{p\to q} = \max\limits_{x} \frac{\norm{Ax}_q}{\norm{x}_p}$. Then  
\begin{align}
\begin{split}
    \maxSE(\mathcal M_{\mathsf{mm}};n) = \widetilde O(\mu \cdot \gamma_{\op 2}(\counting)  ) \quad  &\text{and} \quad \meanSE(\mathcal M_{\mathsf{mm}};n) = O(\mu \cdot \gamma_{\op F}(\counting) ), ~~ \text{where} \\
    \gamma_{\op 2}(A) := \min_{A = LR} \set{\norm{L}_\row \norm{R}_\col}  \quad &\text{and} \quad 
    \gamma_{\op F}(A) := \min_{A = LR} \set{\tfrac{1}{\sqrt{n}}\norm{L}_{\op F}\norm{R}_\col}.
\end{split}
\end{align}

For a specific instantiation of the matrix mechanism, $\mathcal M_{\mathsf{mm}}(\counting,x)=L(Rx+z)$, using  factorization, $\counting=LR$, we denote the maximum squared error by $\text{MaxSE}(L,R)$ and the mean-squared error of the mechanism by $\text{MeanSE}(L,R)$ for preserving $1$-Gaussian differential privacy.

While these $\gamma$-norms provide a principled measure of error of the matrix mechanism,
computing {the optimal factorization} requires solving a computationally expensive semidefinite program (SDP)~\cite{haagerup1980decomposition,  henzinger2023almost}. To address this, a number of explicit factorizations have been proposed that approximate $\gamma_{\op F}(\counting)$~\cite{dwork2010differentially,henzinger2025improved, henzinger2023almost} and $\gamma_{\op 2}(\counting)$~\cite{bennett1977schur,
dvijotham2024efficient,dwork2010differentially,  henzinger2022constant,henzinger2023almost,henzinger2025improved} without relying on SDP solvers. 

Unfortunately, none of these factorizations perform best for both of these metric. For example, the \textit{group algebra factorization}~\cite{henzinger2025improved} gives a better bound for MaxSE error but performs worse than the \textit{square root factorization}~\cite{henzinger2023almost} for MeanSE. See Figure~\ref{fig:factorization_errors} for a plot of the MaxSE and MeanSE error minus $\tfrac{\log (n)}{\pi}$ for (among others) these  factorizations.

\begin{figure}[t]
    \centering
    \caption{\textbf{Left:} the figure shows the MaxSE error for different matrix sizes ($n$). The plot represents the additive difference with $\tfrac{\log (n)}{\pi}$. We also display the optimal factorization value achieved by Nuclear Norm maximization (see Theorem 9 from \cite{lee2008direct}). A lower bound is suggested by \cite{matouvsek2020factorization}. \textbf{Right:} the figure shows the MeanSE error, with the optimal value computed by the corresponding SDP program (see Figure 2 in \cite{henzinger2023almost}). The lower bound in this plot is also suggested by \cite{henzinger2023almost}. In this case, the Normalized Square Root factorization yields lower errors than the standard Square Root factorization and is much closer to the optimal values. }    \includegraphics[width=0.45\linewidth]{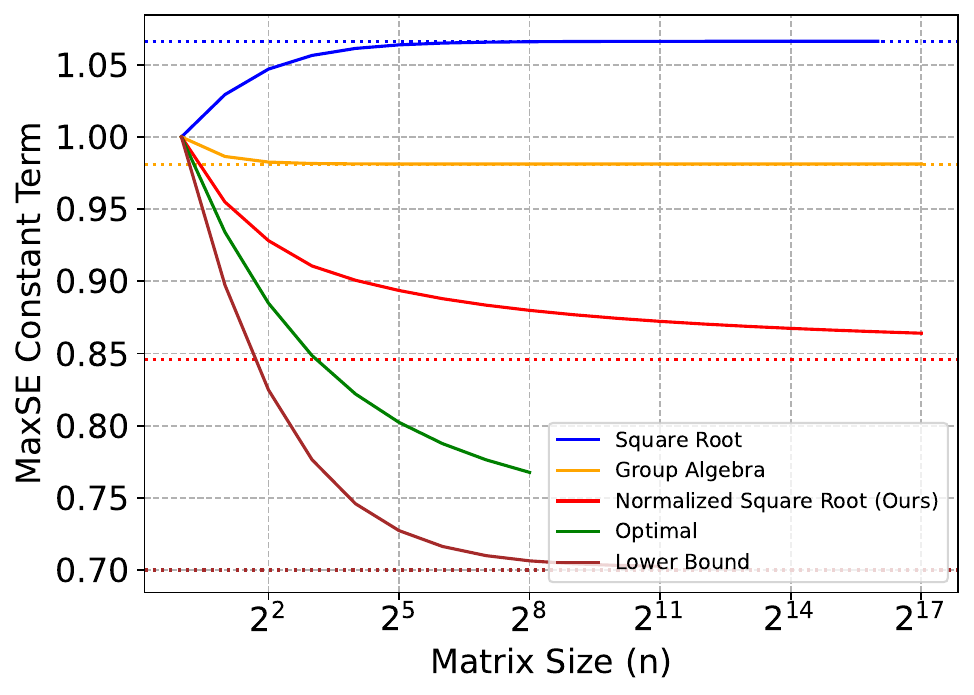}
    \includegraphics[width=0.45\linewidth]{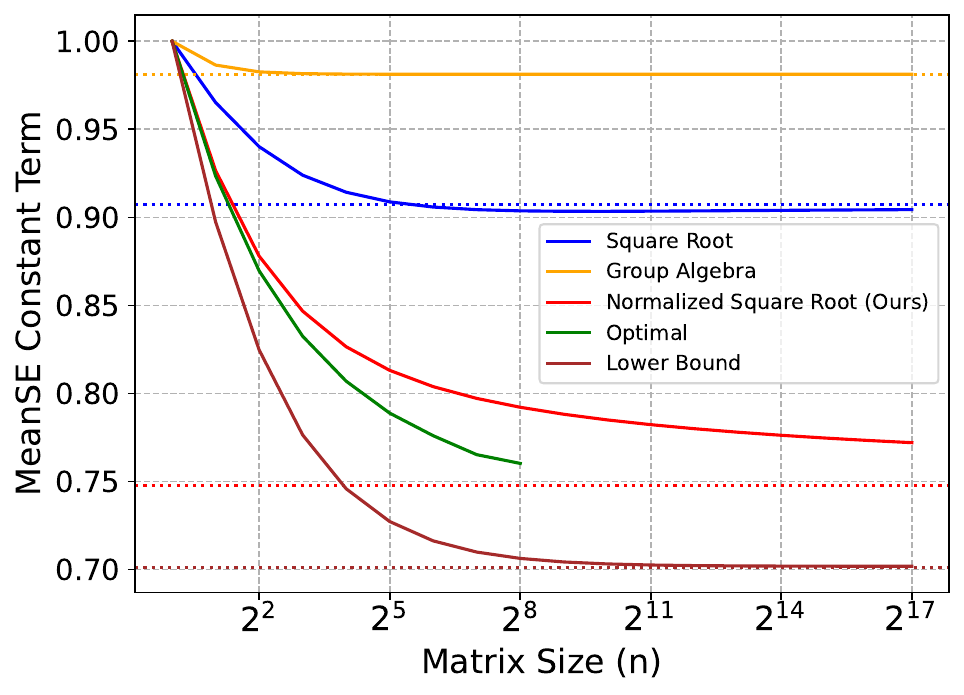}
%    \vspace{-3mm}
    \label{fig:factorization_errors}
\end{figure}

The main contribution of this paper is to provide an explicit factorization that achieves the best known bounds under both MeanSE and  MaxSE metric for $\counting$. As these metrics are proportional to $\gamma_{\op F}(\counting)$ and $\gamma_{\op 2}(\counting)$, we state our result in terms of these norms.

\begin{theorem}
    [Main Result]
    \label{thm:mainTheorem}
    There exists an explicit factorization $LR=\counting$ such that 
    \begin{align*}
    \begin{split}
        0.701 + \frac{\log(n)}{\pi} &\leq \gamma_{\op 2}(\counting) \leq \norm{L}_{\row}\norm{R}_\col \leq \MaxSEUpperBoundExplicit\\
        0.701 + \frac{\log(n)}{\pi} &\leq \gamma_{\op F}(\counting) \leq \frac{1}{\sqrt{n}}\norm{L}_{\op F}\norm{R}_\col \leq \MeanSEUpperBoundExplicit.
    \end{split}
    \end{align*}
\end{theorem}

\subsection*{\texorpdfstring
  {Comparison with prior work for the $\gamma_{\op 2}$-norm}
  {Comparison with prior work for the γop2-norm}}
 The best known bound for $\gamma_{\op 2}(\counting)$ prior to this work was achieved by the group algebra factorization of Henzinger and Upadhyay~\cite{henzinger2025improved}. 
 %They gave an explicit factorization that 
 It was the first to match the non-constructive bound of Mathias~\cite{mathias1993hadamard}, i.e., it achieves $\gamma_{\op 2}(\counting) \leq 1+ \frac{\log(n)}{\pi}$. We show that their factorization actually achieves a somewhat better bound by a careful analysis. That is, for their factors $L_{\mathsf{HU}}R_{\mathsf{HU}}=\counting$, we show in Section~\ref{sub:group_algebra}  that  
\begin{align}
\norm{L_{\mathsf{HU}}}_{2\to \infty} \norm{R_{\mathsf{HU}}}_{1 \to 2} = \frac{1}{2} + \frac{1}{\pi}\paren{\gamma + {\frac{1}{ \pi}}\log\paren{\frac{8}{\pi}}} +  \frac{\log(n)}{\pi} + o(1) \approx 0.981 +  \frac{\log(n)}{\pi},
\label{eq:HUBound}    
\end{align}
 where $\gamma \approx 0.577$ \text{is the Euler-Mascheroni constant}.
Since we show equality, this exactly analyzes the performance of the algorithms in~\cite{henzinger2025improved}, i.e., it is not possible to achieve a better bound with a different analysis. Additionally, it follows that our result improves their bound by {at least} $0.14$. 

We also give the first exact bound for the square root factorization~\cite{bennett1977schur,henzinger2022constant}, namely, we show  in Section~\ref{sec:square_root_factorization} that their factors $L_{\mathsf{FHU}}R_{\mathsf{FHU}}=\counting$ satisfy  
\begin{align}
\norm{L_{\mathsf{FHU}}}_{2\to \infty} \norm{R_{\mathsf{FHU}}}_{1 \to 2} = \frac{\gamma + \log 16}{\pi} - \epsilon_n + \frac{\log (n)}{\pi} \approx 1.066 + \frac{\log (n)}{\pi},\; \text{where} \quad 0 \le \epsilon_n \le \frac{1}{5n}.
\label{eq:FHUBound}    
\end{align}

\begin{table}[t]
    \centering
    \caption{Comparison of our results for $\gamma_2(\counting)$. In the table,     $\gamma$ is the Euler-Mascheroni constant, $\log(\cdot)$ denotes the natural logarithm. All the expressions marked by * have an additional additive $o(1)$ term. Our improvements are in {\color{blue}blue}.}
    \renewcommand{\arraystretch}{1.8}
    \begin{tabular}{|c|c|c|c|}
    \hline
        & Lower bound & Upper Bound & \makecell{Explicit \\ Factors?} \\ \hline 
        Kwapie\'n \& Pe{\l}czy\'nski~\cite{kwapien1970main} & $c \log_2(n)$ & $c_*\log_2(n)$ & $\times$ \\ \hline 
        Mathias~\cite{mathias1993hadamard} & $
        \frac{1}{2} + \frac{1}{2n} + \tfrac{\log(n)}{\pi}$  & ${1} + \frac{\log(n)}{\pi}$ & $\times$ \\ \hline

        Dwork et al.~\cite{dwork2010differentially} & $-$ & $\log_2(n) + \sqrt{\log_2(n)}$ & $\surd$ \\ \hline 

        Matou{\v{s}}ek et al.~\cite{matouvsek2020factorization} & $\frac{1}{ \pi}\paren{\log\paren{\frac{2n+1}{3} } + 2}$ & $-$ & $\times$ \\ \hline 

        {\color{blue} Our improvement of}~\cite{matouvsek2020factorization}* & {\color{blue} $ \underbrace{\tfrac{1}{\pi}\paren{\gamma + \log\paren{\tfrac{16}{ \pi}}}}_{\approx 0.701} + {\frac{\log(n)}{\pi}} $} & $-$ & $\times$ \\ \hline 
        \makecell{Bennett~\cite{bennett1977schur} \\ Fichtenberger et al.~\cite{henzinger2022constant}} & $-$ & {${1 + {\frac{\gamma}{  \pi}}} + \frac{\log(n)}{\pi} \approx 1.184 + \frac{\log(n)}{ \pi}$} & $\surd$ \\ \hline 

        Dvijotham et al.~\cite{dvijotham2024efficient}& $-$ & {$(1+o(1))\paren{1 + {\frac{\gamma}{\pi}+ \frac{\log(n)}{ \pi}}}$} & $\surd$ \\ \hline

        {\color{blue} Our improvement of}~\cite{henzinger2022constant} & $-$ & {\color{blue} $\underbrace{\tfrac{1}{\pi}\paren{\gamma + \log 16}}_{\approx 1.066}+ \frac{\log(n)}{\pi} $} & $\surd$ \\ \hline

        Henzinger \& Upadhyay~\cite{henzinger2025improved} & $-$ & $1 + \frac{\log(n)}{\pi}$ & $\surd$ \\ \hline 

        {\color{blue} Our improvement of}~\cite{henzinger2025improved}* & $-$ & {\color{blue} $\underbrace{\tfrac{1}{2} + \tfrac{1}{\pi}\paren{\gamma + {\tfrac{1}{ \pi}}\log\paren{\tfrac{8}{\pi}}}}_{\approx 0.981} +  \frac{\log(n)}{\pi}$} & $\surd$ \\ \hline 
        
        {\color{blue} Our factorization}* & $-$ & {\color{blue}$\underbrace{\tfrac{1}{\pi}\paren{\gamma + \log 8}}_{\approx 0.845} + \frac{\log(n)}{\pi}$} & $\surd$ \\ \hline
        
    \end{tabular}
    \label{tab:resultsGamma2norm}
\end{table}

Note that even when $n=2$, the best possible bound for 
%Fichtenberger et al.
the factorization ~\cite{henzinger2022constant} is $0.06$ worse than ours and gets increasingly worse as $n$ increases. We note that {the bound on $\gamma_{\op 2}(\counting)$} achieved by the square root factorization 
%of Fichtenberger et al.~\cite{henzinger2022constant}
is  worse than the bound of the group algebra factorization
%Henzinger and Upadhyay
~\cite{henzinger2025improved} for all $n>2$,  as can also be seen by our numerical simulation (Figure~\ref{fig:factorization_errors}). In fact, it converges to $1.0663 + \frac{\log(n)}{\pi}$ as soon as $n\approx 100$. 

On the lower bound side, our tighter analysis also allows us to improve the previously computed bound. More precisely, Matou\u{s}ek et al.~\cite{matouvsek2020factorization} set the dual certificates, $u=v={ \frac{\mathbf 1}{\sqrt{n}}}$ (i.e., the normalized all-ones vector) in the following dual formulation of $\gamma_{\op 2}(\counting)$:
\[
\gamma_{\op 2}(\counting) = \max_{v, u, \|u\|,\|v\|=1} \norm{\counting \bullet uv^\top}_*,
\]
Here $A\bullet B$ denotes the \textit{Schur product} and $\norm{\cdot} _*$ denotes the \textit{nuclear norm}. 
Dvijotham et al.~\cite{dvijotham2024efficient} showed that Matou\u{s}ek et al.~\cite{matouvsek2020factorization}'s dual certificate satisfies  $\gamma_{\op 2}(\counting)\geq \frac{1}{\pi}\paren{\log((2n+1)/3) + 2} \approx 0.507 + \frac{\log(n)}{\pi}$. This gives a gap of $\approx 0.6767$ from the then-known best known upper bound~\cite{henzinger2022constant}. They further report that empirically, the  gap between the lower bound using the dual certificate of Matou\u{s}ek et al.~\cite{matouvsek2020factorization} and 
the square root factorization
%Fichtenberger et al.~\cite{henzinger2022constant} 
is less than $0.365$. 

{Using a more rigorous analysis, we show that the dual certificate of Matou\u{s}ek et al.~\cite{matouvsek2020factorization} achieves a better lower bound of  $\gamma_{\op 2}(\counting) \geq 0.701 + \frac{\log(n)}{\pi} + o(1)$}. {Since $\gamma_{\op 2}(\counting)\leq \norm{L_{\mathsf{FHU}}}_{2\to \infty} \norm{R_{\mathsf{FHU}}}_{1 \to 2}$, the lower bound in Theorem~\ref{thm:mainTheorem}} in conjunction with eq.~\eqref{eq:FHUBound} shows the validity of the empirical observation of~\cite{dvijotham2024efficient}, i.e., the gap between these two bounds is $0.365$.  Note that, the bound in eq.~\eqref{eq:FHUBound} is an equality; so it also shows that improving the constant in the lower bound would require a different dual certificate than \cite{matouvsek2020factorization}.
We present these comparisons succinctly in Table~\ref{tab:resultsGamma2norm}.

\subsection*{\texorpdfstring
  {Comparison with prior work for the $\gamma_{\op F}$-norm}
  {Comparison with prior work for the γopF-norm}}
 The best known upper and lower bound for $\gamma_{\op F}(\counting)$ prior to this work was for the square root factorization and it was given by Henzinger et al.~\cite{henzinger2023almost}. They showed the following:
 \[
 \frac{1}{\pi}\paren{2 +{\log\paren{\frac{(2n+1)}{3}}}} \leq \gamma_{\op F}(\counting) \leq 1 +\frac{1}{ \pi}\paren{\gamma+\log(n)}.
 \]
 
 We improve their lower bound such that the constant term is $\approx 0.2$ better. As a consequence, there is just a $0.047$ additive gap between the lower bound and the upper bound achieved by our explicit factorization (Theorem~\ref{thm:mainTheorem}). 
 We also show that the square root factorization achieves better $\gamma_F$ upper bound than the group algebra factorization~\cite{henzinger2025improved}. In fact, we show in Section~\ref{sec:square_root_factorization} that the factorization $L_{\mathsf{FHU}}R_{\mathsf{FHU}} =\counting$ fulfills the following exact relation:
\begin{align}
\text{MeanSE}(L_{\mathsf{FHU}}, R_{\mathsf{FHU}}) = \frac{1}{\sqrt{n}}\norm{L_{\mathsf{FHU}}}_{\op F} \norm{R_{\mathsf{FHU}}}_{1 \to 2} = \frac{\gamma + \log 16}{\pi} - \frac{1}{2\pi} + \frac{\log (n)}{\pi} + o(1).
\label{eq:HUUBound}    
\end{align}
Compared to eq \eqref{eq:HUBound}, this is better by an additive $0.08$ term and {only} an ${\frac{1}{ 2\pi}}$ additive term worse than Theorem~\ref{thm:mainTheorem}.  Table~\ref{tab:resultsGammaFnorm} enumerates these results.

\begin{table}[t]
    \centering
    \caption{Comparison of our results for $\gamma_{\op F}(\counting)$. In the table,     $\gamma$ is the Euler-Mascheroni constant,   $\log(\cdot)$ denotes the natural logarithm. All the expressions marked by * have an additional additive $o(1)$ term. Our improvements are in {\color{blue}blue}.}
    \renewcommand{\arraystretch}{2}
    \begin{tabular}{|c|c|c|c|}   
    \hline
        & Lower bound & Upper Bound & \makecell{Explicit \\ Factors?} \\ \hline

        Henzinger et al.~\cite{henzinger2023almost} & $\frac{1}{\pi}\paren{2 +{\log\paren{\frac{(2n+1)}{3}}}}$ & {$1 +{\frac{\gamma}{\pi}}  + \frac{\log(n)}{ \pi}\approx 1.184 + \frac{\log(n)}{\pi}$} & $\surd$ \\ \hline 
        
        {\color{blue} Our improvement of}~\cite{henzinger2023almost} & {\color{blue} $ \underbrace{\tfrac{1}{\pi}\paren{\gamma + \log\paren{\tfrac{16}{ \pi}}}}_{\approx 0.701} + \frac{\log(n)}{\pi} $} & {\color{blue} $\underbrace{\tfrac{1}{\pi}\paren{\gamma + \log (16) - \tfrac{1}{2}}}_{\approx 0.907} + \frac{\log (n)}{\pi}$} & $\surd$ \\ \hline 

        Henzinger \& Upadhyay~\cite{henzinger2025improved} & $-$ & $1 + \frac{\log(n)}{\pi}$ & $\surd$ \\ \hline 
        {\color{blue} Our improvement of}~\cite{henzinger2025improved}* & $-$ & {\color{blue} $\underbrace{\tfrac{1}{2} + \tfrac{1}{\pi}\paren{\gamma + {\tfrac{1}{ \pi}}\log\paren{\tfrac{8}{\pi}}}}_{\approx 0.981} +  \frac{\log(n)}{\pi}$} & $\surd$ \\ \hline
        {\color{blue} Our factorization}* & $-$ & {\color{blue}$\underbrace{\tfrac{1}{\pi}\paren{\gamma + \log (16) - 1} }_{\approx 0.748} + \frac{\log(n)}{ \pi}$} & $\surd$ \\ \hline
        
    \end{tabular}
    \label{tab:resultsGammaFnorm}
\end{table}

\subsection*{Implications in private training} 
Our factorization uses the square root factorization $L_{\mathsf{FHU}}$ $R_{\mathsf{FHU}}$, then normalizes the columns of $R_{\mathsf{FHU}}$ to compute $\widetilde R$, and finally compute $\widetilde L=A\widetilde R^{-1}$. Column normalization is a heuristic used in the current deployment of private training of a neural network (see \cite{choquette2022multi}). This is motivated by the empirical observation that, for $n$ up to $4096$, the right factor of the optimal factorization with respect to the MeanSE($L,R$) has normalized columns (also, see the discussion on~\cite[page 62]{pillutla2025correlated}). However, from the empirical observation, one cannot infer for which matrix factorization we should normalize the columns or whether this phenomenon exists for all $n$ or only for some bounded values of $n$. Our paper takes the first theoretical step in this direction. 
In particular, even though MeanSE($L,R$) has been traditionally used in several works~\cite{choquette2022multi,choquette2023correlated,mcmahan2022private}, a recent work shows that it might not be the only statistic defining the accuracy of private training~\cite[Figure 3]{ganesh2025design} {and Dvijotham et al.~\cite{dvijotham2024efficient} suggest that MaxSE($L,R$) might also be very useful.}

\subsection*{Implication in Other Areas of Mathematics.} 
On its own, factorization norms are a central question in functional analysis and operator algebra since the first discussion in  Schur~\cite{schur1911bemerkungen}, where he introduced Schur multipliers. In particular, the matrix $\counting$ has played a special role  since Kwapie\'n and Pe{\l}czy\'nski~\cite{kwapien1970main} showed its deeper connection with various seemingly unrelated problems, like absolutely summing operators and absolute convergence of series in $\ell_1$. Since then, $\gamma_{\op 2}(\counting)$  has found applications in non-commutative matrix algebra~\cite{junge2005best}, symplectic capacity~\cite{gluskin2019symplectic}, compact operators~\cite{aleksandrov2023triangular,kato1973continuity}, other absolute summing problems~\cite{gordon1974absolutely}, etc.

\subsection{Our Techniques for Upper Bound} 
To understand the intuition behind our factorization, we revisit the error metric and the factorizations that previously achieved the best-known results: the \emph{square root factorization}~\cite{bennett1977schur,henzinger2022constant} and the \emph{group algebra factorization}~\cite{henzinger2025improved}. Both possess distinctive structural properties, which in turn explain why each is well-suited to a particular error measure (MeanSE or MaxSE). For example, the group algebra factorization, $L_{\mathsf{HU}} R_{\mathsf{HU}} = \counting$, has the property that every row of $L_{\mathsf{HU}}$ and every column of $R_{\mathsf{HU}}$ have the same $\ell_2$-norm. This uniformity is ideal for MaxSE: no row of $L_{\mathsf{HU}}$ or column of $R_{\mathsf{HU}}$ dominates, so their contributions are evenly balanced. In contrast, the square root factorization gives disproportionate weight to the last row and the first column. While this structure is suboptimal for MaxSE, it allows the error to be averaged across all rows, which significantly lowers the MeanSE error. Thus, the square root factorization benefits from the change of metric, whereas the group algebra factorization does not.

\subsection*{Finding the factors}

Our starting point for our improvement is the square root factorization. Unlike the group algebra factorization, normalizing the columns in this setting scales the $\ell_2$-norms of the corresponding {\em rows} differently. In fact, this normalization accentuates the desired phenomenon: the disbalance between the row norms becomes larger, with the $\ell_2$-norm of the rows of the left factor increasing more rapidly than in the unnormalized square root factorization. Crucially, however, this growth remains controlled. The main technical challenge in our analysis is to bound this rate of increase as tightly as possible.  

To illustrate concretely that normalizing helps, let us consider the simplest non-trivial case of $n=2$. Normalizing the column norm of $\counting^{1/2}$ to obtain $\widetilde{C}$ gives 
\[
\widetilde C = \begin{pmatrix}
    \tfrac{2}{\sqrt{5}} & 0 \\
    \tfrac{1}{\sqrt{5}} & 1
\end{pmatrix}
 \quad \text{and} \quad 
\widetilde{B} = \counting \widetilde C^{-1} = \begin{pmatrix}
    \tfrac{\sqrt{5}}{2} & 0  \\
    \tfrac{\sqrt{5}-1}{2} & 1 
\end{pmatrix}.
\]
This factorization yields $\gamma_{\op 2}(\counting) \leq 1.1755$, whereas, for $n=2$,  Henzinger and Upadhyay~\cite{henzinger2025improved} obtain $\gamma_{\op 2}(\counting) \le 1.25$. A similar calculation can be carried out for $\gamma_{\op F}(\counting)$, showing that normalizing the columns of $R$ indeed leads to an improvement.  

This motivates the following construction: %, which we call the {\normalized square root (NSR) factorization:

\begin{definition}[Normalized Square Root Factorization]
Let $C$ be the square root of $\counting$, i.e., $C^2=\counting$, $C_{:,i}$ denote the $i$-th column of $C$, and $D = \operatorname{diag}\left(\norm{C_{:,1}}_2, \cdots, \norm{C_{:,n}}_2 \right)$ denote the diagonal matrix with entries $d_1, \cdots, d_n$ containing the column norms of $C$. 
Then the \emph{Normalized Square Root (NSR)} factorization of a matrix $\counting$ is given by $\counting = \widetilde{B}\widetilde{C}$, where
\begin{equation}
    \widetilde{C} = C D^{-1} \qquad \text{and} \qquad \widetilde{B} = \counting \widetilde{C}^{-1} = \counting D C^{-1}.
\label{eq:NSRFactorization}
\end{equation}
\end{definition}

%- Main properties of FHU and HU.

\subsection*{Proving Upper Bounds in Theorem~\ref{thm:mainTheorem}}
Prior to this work, it was known that $\gamma_{\op 2}(\counting) = c_n+ \frac{\log(n)}{\pi}$ for some $c_n \in (1/2,1)$. By construction, NSR ensures that $\|\widetilde{C}\|_{1\to 2}=1$. Therefore, to bound $\gamma_{\op 2}(\counting)$ and $\gamma_{\op F}(\counting)$, all we need is to compute the $j$-th row norm of $\widetilde B$ for $1\leq j \leq n$. Let $\widetilde{B}_{j,:}$ denote the $j$-th row of $\widetilde{B}$. We aim to show that
\begin{align}
\max_j\norm{\widetilde{B}_{j,:}}_2 \leq   \frac{1}{\pi} (\gamma+ \log(8)) + \frac{\log(n)}{ \pi} +o(1). 
\label {eq:boundwidetildeB_j_norm}   
\end{align}

A crucial aspect in proving bounds in previous works on explicit factorizations is that the entries of the factors could be controlled explicitly. For instance, in the square root factorization $\counting=C^2$, the matrix $C$ is lower triangular Toeplitz with generating function of $(1-x)^{-1/2}=r_0+r_1x+\cdots$, and its diagonal entries are given by the coefficients $r_i$. By Wallis’ inequality (Lemma~\ref{lem:WallisInequality}), these coefficients admit tight bounds, which in turn yield bounds on $\gamma_{\op 2}(\counting)$ and $\gamma_{\op F}(\counting)$. Similarly, the group algebra factorization admits an explicit representation for each entry (Lemma~\ref{lem:dft_decomposition}). {Once we have a bound on every entry of the factors, estimating both $\gamma_{\op 2}(\counting)$ and $\gamma_{\op F}(\counting)$ in the case of square root factorization~\cite{bennett1977schur, henzinger2022constant} and group algebra factorization~\cite{henzinger2025improved} is straightforward due to their structural property: the $\ell_2$ norm of rows forms an increasing sequence in square-root factorization and they are equal for the group algebra factorization}

We use the same approach and first bound the entries of the matrix $\widetilde{B}$ (denoted by $\widetilde{B}_{jk}$ for $(j,k)$-entry of $\widetilde{B}$). However, unlike the square root factorization and the group algebra factorization, the normalized square root factorization does not have the norms of rows forming a monotonically non-decreasing sequence. Therefore, we need to find the row that has the maximal norm. %We next discuss our approach to performing both of these tasks.

\subsubsection*{\underline{Upper Bound on Entries of the Matrix $B$}} 
Since $C=\counting^{1/2}$, the diagonal entries of $C^{-1}$ can be expressed via the generalized binomial expansion of $\sqrt{1-x}=\widetilde r_0+\widetilde r_1x+\cdots$. Let $d_i=\norm{C_{:,i}}_2$ ($1\leq i \leq n$) be the diagonal entries of $D$. Since $\widetilde{B}=\counting DC^{-1}$, we obtain
\[
 \widetilde{B}_{jk}=\sum_{t=0}^{j-k} \widetilde r_t d_{k+t} \quad \text{for}\quad 1 \leq j,k \leq n, \quad \text{where}\quad d^2_i = \sum\limits_{t = 0}^{n - i}r_t^2.
\]

The key idea is to reorganize the sum by exploiting structure in the coefficients $\widetilde r_t$. Since these are precisely the entries of $C^{-1}$ and $\counting C^{-1}=C$ {(since $C$ is the square root factorization of $\counting$)}, we have the identity $\sum_{t=0}^j \widetilde{r}_t = r_j$. Applying this via summation by parts allows us to rewrite $\widetilde B_{jk}$ in terms of differences of consecutive $d$’s. Moreover, the sequence $(d_i)_{i \geq 1}$ satisfies $d_j^2 - d_{j+1}^2 = r_{n-j}^2$, 
which lets us further express these differences explicitly. Combining both ingredients gives

\begin{equation} 
\widetilde{B}_{jk}
   = \sum_{t=0}^{j-k} \widetilde r_t d_{k+t} = d_j r_{j - k} + \sum_{t = 0}^{j - k - 1} (d_{t + k} - d_{t + k + 1}) r_t = d_j r_{j - k} + \sum_{t = 0}^{j - k - 1} \frac{r_{n - k - t}^2}{d_{t + k} + d_{t + k + 1}} r_t.  
\label{eq:B_j_k_introduction}    
\end{equation}

This decomposition is the starting point for our refined analysis: the first term captures the dominant contribution, while the second sum of positive terms can be carefully bounded. This is the challenging part. 
A first natural attempt might be to upper bound the sum in $\widetilde{B}_{jk}$ by lower bounding the denominator, e.g.\ using $d_{t+k+1} + d_{t+k}\geq 2d_j$, {since $d$ is a decreasing sequence.} 

However, this yields the wrong asymptotics ($\log^{3/2}(n)$ instead of $\log(n)/\pi$), illustrating the subtlety of the proof. Standard integral bounds are also ineffective here because of the combination of $r_t$ and $d_{t+k}$ terms. By splitting into two different regimes of $t$, with $t > \left\lfloor \tfrac{2(j - k)}{3}\right\rfloor$, we can eliminate $r_t$ by bounding it with the largest value on the interval. The choice of the constant $\tfrac{2}{3}$ for the split here is important, as this constant should be strictly larger than $\tfrac{1}{2}$; otherwise, the error term would start to dominate asymptotically, and we would not get the correct constant. Carrying out this splitting, the contribution from the regime $t > \lfloor \tfrac{2(j-k)}{3} \rfloor$ can be bounded as
\begin{equation}
    \sum\limits_{t = \left\lfloor \tfrac{2(j - k)}{3}\right\rfloor + 1}^{j - k - 1} \frac{r_{n - k - t}^2}{d_{t + k} + d_{t + k + 1}} r_t
    \le \frac{\sqrt{3/2}}{\pi} \cdot \frac{\sqrt{\log (n - k)} - \sqrt{\log (n - j + 1)}}{\sqrt{n - k}}.
    \label{eq:firstsplit}
\end{equation}

The remainder can be bounded using the rough inequality $d_t + d_{t + k + 1} \ge 2$:
\begin{equation}
    \sum\limits_{t = 0}^{\min\left\{j - k - 1, \left\lfloor \tfrac{2(j - k)}{3}\right\rfloor \right\}} \frac{r_{n - k - t}^2}{d_{t + k} + d_{t + k + 1}} r_t
    \le \frac{1}{\sqrt{n - k}}.
    \label{eq:secondsplit}
\end{equation}

Then, an upper bound on $\sum_k\widetilde{B}_{jk}^2$ introduces six terms, five of which involves either eq.~\eqref{eq:firstsplit} or eq.~\eqref{eq:secondsplit}. In Section~\ref{sec:our_factorization}, we show that the effect of these five terms is dominated by the leading term $\sum_{k} d_j^2r_{j-k}^2$. 
The intuition behind this is that, terms involving  eq.~\eqref{eq:firstsplit}, when we square the terms and sum over $k$, the first bound contribute a $\tfrac{3}{4\pi^2}\log^2 n$ term when $j$ is close to $n$, but the constant would be worse than $\tfrac{1}{\pi^2}$, which will later be shown to be optimal. The second bound can contribute a $\log n$ term when $j$ is close to $n$, but when $n - j$ is large, they contribute only $o(\log n)$, which is roughly the behavior we wish to achieve. We prove this rigorously in Section~\ref{sec:our_factorization}. 
%Now while we know the exact value of $\widetilde r_t$ for $0 \leq t \leq j-k$, we only have a bound on $d_{k+t}^2$, which is sum of squares of $r_0,\cdots, r_{k+t-1}$. Here, we can use the fact that $r_i$ satisfies the Wallis' inequality. Therefore, one can compute a tight bound on every entry of $\widetilde B_{j,k}$. While this seems straightforward, ensuring that we are tight up to the additive constants requires extra care. 

\subsubsection*{\underline{Finding row $j$ with maximal row norm}} 
In eq.~\eqref{eq:firstsplit} and eq.~\eqref{eq:secondsplit}, we have bounded the additional sum in \eqref{eq:B_j_k_introduction} and shown that its contribution is negligible in the regime of $j$ where we expect the optimum; we now analyze the main term $d_j r_{j-k}$.
For the {leading term, $d_j r_{j-k}$, in eq.~\eqref{eq:B_j_k_introduction}}, from~\cite{bennett1977schur, henzinger2022constant} we know that $d_j \sim \sqrt{\log(n - j + 1)/\pi}$, and $r_{j - k}$ satisfies Wallis' inequality (Lemma~\ref{lem:WallisInequality}). Hence, the individual values can be tightly bounded. Moreover, the sum of their squares, corresponding to the $j$-th squared row norm of the matrix $\tilde{B}$, admits a convenient closed form { (note that we only need to worry about the leading term as we have shown the other terms are strictly dominated by it)}:
\[
    \sum_{k = 1}^{j} (d_j r_{j - k})^2 
      = d_j^2 \sum_{k = 1}^{j} r_{j - k}^2 
      = d_j^2 d_{n - j + 1}^2.
\]

To show improved constants, however, we need a more precise estimate of $d_j$, including the next asymptotic term. We prove that there exists a constant $\alpha_{\infty} \approx 1.0663$ such that
\begin{equation}
    d_{n - j + 1}^2 = \alpha_{\infty} + \frac{\log(j)}{\pi} - \epsilon_{j}, \quad \text{where} \quad 0 < \epsilon_{j} \le \tfrac{1}{5j}.
\end{equation}

Consequently, the contribution of the main term in \eqref{eq:B_j_k_introduction} to the squared MaxSE error is
\begin{equation}
    d_j^2 d_{n - j + 1}^2 
       = \frac{\log(n - j + 1)\log(j)}{\pi^2} 
         + \frac{2\alpha_{\infty}}{\pi} \log(n) 
         + o(\log(n)).
\end{equation}

The expression $\log(n - j + 1)\log(j)$ achieves its maximum when $j \sim n/2$, giving
\begin{equation}
    \max_j d_j^2 d_{n - j + 1}^2 
       = \frac{\log^2(n)}{\pi^2} 
         + \frac{2(\alpha_{\infty} - \log(2)/\pi)}{\pi} \log(n) 
         + o(\log(n)).
\end{equation}

Note that the constant $\alpha_{\infty}$ had already been computed in an unrelated context in mathematical analysis, namely the study of Landau constants (also see Section~\ref{sec:landauConstant}). In fact, Watson \cite{watson1930constants} showed that
\[
\alpha_{\infty} = \frac{\gamma + \log(16)}{\pi},
\]
where $\gamma$ is Euler’s–Mascheroni constant. Verifying that the error term is negligible in the regime $j \sim n/2$ and taking the square root completes the proof of the MaxSE bound.

In short, the key takeaways from our proof can be distilled into three points:
\begin{enumerate}
    \item The first term in eq.~\eqref{eq:B_j_k_introduction} dominates the sum.  
    \item A precise asymptotic for the values $d_j$ is required.  
    \item The row achieving the maximum $\ell_2$ norm lies at $j \sim n/2$.  
\end{enumerate}

While the first two points are essential for bounding both $\gamma_{\op 2}(\counting)$ and $\gamma_{\op F}(\counting)$, the third point is crucial specifically for $\gamma_{\op 2}(\counting)$. Proving (and applying) points (1) and (3) requires special care, since the non-dominant term also has $\Theta(\log n)$ growth, albeit with a smaller constant than $1/\pi$. A loose analysis can therefore lead to incorrect values both in the leading term $\log n$ and in the constant term we aim to improve. For the $\gamma_{\op F}(\counting)$ bound, we need to compute the average of the squared row norms. The main challenge lies in determining the constant in the leading term, while the remaining terms contribute only $o(\log n)$, which does not affect the asymptotics.\\

The proofs of the known factorizations can be found in Sections \ref{sec:square_root_factorization} and \ref{sub:group_algebra}. The square root factorization \eqref{eq:HUUBound} follows from the analysis of $d_{n - j+1}^2$, considering both the maximum and the average over $j$. The proof of the group algebra factorization \eqref{eq:HUBound} uses the same technique as the lower-bound proof presented in this paper. In particular, we need to average the sequence $\frac{1}{\sin(j\pi/n)}$, which diverges. After subtracting a harmonic sum, the remainder takes the form of a Riemann sum, and its limit is identified with the corresponding integral, which we compute explicitly.

\section{Proof of Lower Bounds in Theorem~\ref{thm:mainTheorem}}
We first state a key lemma (proved at the end of the section) required in our lower bound proof. 

\begin{lemma}
\label{lem:boundonG(n)}
Let $\gamma \approx 0.57$ denote the Euler-Mascheroni's constant. For $n \in \mathbb N$,  
\begin{align}
G(n) := \frac{1}{n} \sum\limits_{\ell = 1}^{n - 1} \frac{1}{\sin\left(\tfrac{\pi \ell}{n}\right)}  = \frac{2}{\pi} \paren{\log(n) + \gamma + \log\paren{\frac{2}{\pi}}} + o(1).
\label{eq:boundonG(n)}    
\end{align} 
\end{lemma}

\subsection{Improving the Lower Bounds}
Now we return to our improvement on the lower bounds. 
The best known lower bound on $\gamma_{\op 2}(\counting)$ is by  Matousek et al. \cite{matouvsek2020factorization} in terms of cosecant function:

\begin{equation*}
    \gamma_{\op 2}(\counting) \ge \frac{\|\counting\|_{*}}{n} = \frac{1}{2n} \sum\limits_{j =1}^{n} \frac{1}{\sin \left(\tfrac{(2j - 1)\pi}{4n + 2}\right)}.
\end{equation*}

The constant in the lower bound above was calculated by Dvijotham et al.~\cite{dvijotham2024efficient}:
\begin{equation*}
    \frac{1}{2n} \sum\limits_{j =1}^{n} \frac{1}{\sin \left(\tfrac{(2j - 1)\pi}{4n + 2}\right)} \ge \frac{\log \left(\tfrac{2n + 1}{3}\right) + 2}{\pi} = \frac{\log (n)}{\pi} + \underbrace{\frac{2 + \log \left(\tfrac{2}{3}\right)}{\pi}}_{\approx 0.507} + o(1).
\end{equation*}

We compute the lower bound more precisely and improve the known constant term as follows:

\begin{lemma}
\label{lem:max_se_lower_bound}
    The optimal value for the $\gamma_{\op 2}(\counting)$ of the prefix sum matrix $\counting$ satisfies the following lower bound:
\begin{equation*}
    \gamma_{\op 2}(\counting) \ge \frac{\|\counting\|_{*}}{n} = \frac{\log (n)}{\pi} + \underbrace{\frac{1}{ \pi}\paren{\log \left( \frac{16}{\pi}\right) + \gamma}}_{\approx 0.701} + o(1).
\end{equation*}
\end{lemma}

\begin{proof}
Since for $x \in (0, \pi)$ we have $\sin(x  + o(1)) = \sin(x)\cos(o(1)) +\cos(x)\sin(o(1)) = \sin(x) + o(1)$, we begin with the following asymptotic expansion:
\begin{align*}
    \frac{\|\counting\|_{*}}{n} = 
    \frac{1}{2n} \sum\limits_{j =1}^{n} \frac{1}{\sin \left(\tfrac{(2j - 1)\pi}{4n + 2}\right)} = \frac{1}{2n} \sum\limits_{j =1}^{n} \frac{1}{\sin \left(\tfrac{(2j - 1)\pi}{4n}\right)} + o(1),
\end{align*}
where we used that the average of $o(1)$ elements is also $o(1)$. 

With this adjusted sum, we concentrate on the summation 
\begin{equation*}
    \frac{1}{2n} \sum\limits_{j =1}^{n} \frac{1}{\sin \left(\tfrac{(2j - 1)\pi}{4n}\right)} = \frac{1}{2n} \sum\limits_{j =1}^{2n} \frac{1}{\sin \left(\tfrac{j\pi}{4n}\right)} - \frac{1}{2n} \sum\limits_{j =1}^{n} \frac{1}{\sin \left(\tfrac{2j\pi}{4n}\right)} = F(2n) - \tfrac{1}{2}F(n),
\end{equation*}
where $F(n)$ is defined as the sum of cosecant functions:

\begin{equation*}
    F(n) := \frac{1}{n} \sum\limits_{j =1}^{n} \frac{1}{\sin \left(\frac{j}{n}\cdot \frac{\pi}{2}\right)}.
\end{equation*}

We compute $F(n)$ as follows:

\begin{equation*}
    F(n) - \frac{H_n}{\pi/2} = \frac{1}{n} \sum\limits_{j = 1}^{n} \left[\frac{1}{\sin \left(\frac{j}{n}\cdot \frac{\pi}{2}\right)} - \frac{1}{\frac{j}{n}\cdot \frac{\pi}{2}}\right] = \frac{2}{\pi}\log\left(\frac{4}{\pi}\right) + o(1),
\end{equation*}
where the sum is evaluated in equation~\eqref{eq:integral_equation} of Lemma~\ref{lem:boundonG(n)}. For the value of $H_n$, we use the standard expansion $H_n = \log (n) + \gamma + o(1)$, where $\gamma \approx 0.57721$ is the Euler–Mascheroni constant. Thus:

\begin{align*}
    F(2n) - \tfrac{1}{2}F(n) &= \frac{\log (2n) + \gamma + \log\left(\frac{4}{\pi}\right)+  o(1)}{\pi/2} + o(1) - \frac{1}{2}\left[\frac{\log (n) + \gamma + \log\left(\frac{4}{\pi}\right) + o(1)}{\pi/2} +  o(1)\right] \\
    &= \frac{\log (n)}{\pi} + {\frac{1}{\pi}\paren{\log \left(\frac{16}{\pi}\right) + \gamma}} + o(1),
\end{align*}
concluding the proof of Lemma~\ref{lem:max_se_lower_bound}.
\end{proof}

For the $\gamma_{\op F}(\counting)$, it was shown in Henzinger et al. \cite{henzinger2023almost} that

\begin{equation*} 
    \gamma_{\op F}(\counting) \ge \frac{1}{\pi} \left(2 + \log \left(\frac{2n + 1}{5}\right) + \frac{\log(2n + 1)}{2n}\right) = \frac{\log (n)}{\pi} + \underbrace{\frac{2 + \log(2/5)}{\pi}}_{\approx 0.344} + o(1).
\end{equation*}

Recall that Henzinger et al.~\cite{henzinger2023almost} achieved their bound by computing a lower bound on $\norm{A}_*$. Therefore, as a direct corollary of Lemma~\ref{lem:max_se_lower_bound}, we have also improved their bound:

\begin{corollary}
For a matrix $\counting \in \{0,1\}^{n\times n}$ defined by eq.~\eqref{eq:countingMatrix}
\begin{equation*}
    \gamma_{\op F}(\counting) \ge \frac{\|\counting\|_{*}}{n} = \frac{\log (n)}{\pi} + \underbrace{\frac{1}{ \pi}\paren{\log \left(\frac{16}{\pi}\right) + \gamma}}_{\approx 0.701} + o(1).
\end{equation*}
\end{corollary}

For completeness, we also analyze the bound in Mathias~\cite[Corollary 3.5]{mathias1993hadamard}, who showed the following lower bound for the MaxSE error:

\begin{equation*}
    \gamma_{\op 2}(\counting) \ge \frac{n +1}{2n^2}\sum\limits_{j = 1}^{n}\frac{1}{\sin \left(\tfrac{(2j - 1)\pi}{2n}\right)}
\end{equation*}

We compute the leading terms in the following lemma:

\begin{lemma}
\label{lem:LB_Mathias}
The lower bound presented in Mathias~\cite[Corollary 3.5]{mathias1993hadamard} satisfy:
\begin{equation*}
    \gamma_{\op 2}(\counting) \ge \frac{n +1}{2n^2}\sum\limits_{j = 1}^{n}\frac{1}{\sin \left(\tfrac{(2j - 1)\pi}{2n}\right)} = \frac{\log (n)}{\pi} + \underbrace{\frac{\log \paren{\tfrac{8}{\pi}} + \gamma}{\pi}}_{\approx 0.481} + o(1)
\end{equation*}
\end{lemma}

\begin{proof}
We first factor the expression as
\[
    \frac{n+1}{2n^2} \sum_{j=1}^{n}\frac{1}{\sin \!\left(\tfrac{(2j-1)\pi}{2n}\right)}
    = \Bigl(1 + \tfrac{1}{n}\Bigr)\Biggl[\frac{1}{2n}\sum_{j=1}^{n}\frac{1}{\sin\!\left(\tfrac{(2j-1)\pi}{2n}\right)}\Biggr] =\Bigl(1 + \tfrac{1}{n}\Bigr) \left[G(2n) - \tfrac{1}{2}G(n)\right],
\]
where $G(n) = \tfrac{1}{n}\sum_{j=1}^{n-1}\tfrac{1}{\sin(\pi j/n)}$.  
By Lemma~\ref{lem:boundonG(n)}, we conclude that
\[
    \frac{n+1}{2n^2} \sum_{j=1}^{n}\frac{1}{\sin\!\left(\tfrac{(2j-1)\pi}{2n}\right)}
    = \Bigl(1+\tfrac{1}{n}\Bigr)\left[\frac{\log n}{\pi} + \frac{\gamma + \log(8/\pi)}{\pi} + o(1)\right].
\]
Since $\tfrac{\log n}{n} = o(1)$, the proof of Lemma~\ref{lem:LB_Mathias} is complete.
\end{proof}

\subsection{Proof of Lemma~\ref{lem:boundonG(n)}}

For convenience, we analyze $G(2n)$ since it has an even argument. The analysis for $G(n)$ with odd $n$ is analogous, up to $o(1)$ terms caused by rounding, in particular $G(n) = G(2\lfloor \tfrac{n}{2} \rfloor) + o(1)$. 

\begin{equation*}
    G(2n) - \frac{2H_{n - 1}}{\pi} = \frac{1}{2n} \sum\limits_{\ell = 1}^{2n - 1} \frac{1}{\sin\left(\tfrac{\pi \ell}{2n}\right)} - \frac{2H_{n - 1}}{\pi} = \frac{1}{2n} +  \frac{1}{n} \sum\limits_{\ell = 1}^{n - 1} \left[\frac{1}{\sin\left(\tfrac{\pi \ell}{2n}\right)} - \frac{1}{\tfrac{\pi \ell}{2n}}\right].
\end{equation*}

This is a Riemann sum, as we split the interval $[0, 1]$ into smaller segments of the form $\frac{\ell}{n}$, the value at $l = 0$ goes to $0$ by continuity. Thus
\begin{equation}
\label{eq:integral_equation}
    \frac{1}{n} \sum\limits_{\ell = 1}^{n - 1} \left[\frac{1}{\sin\left(\tfrac{\pi \ell}{2n}\right)} - \frac{1}{\tfrac{\pi \ell}{2n}}\right] \xrightarrow{} \int\limits_{0}^1 \left[\frac{1}{\sin(\pi y/2)} - \frac{1}{\pi y/2}\right]\, \mathsf dy = \frac{2}{\pi} \int\limits_{0}^{\pi / 2} \left[\frac{1}{\sin x} - \frac{1}{x}\right]\, \mathsf dx
\end{equation}
by the change of variable $x=\pi y/2$, which results in $2 \mathsf dx =\pi \mathsf dy$.  

The integral can be computed explicitly and gives
\begin{equation*}
    \int\limits_{0}^{\pi / 2} \left[\frac{1}{\sin x} - \frac{1}{x}\right]\, \mathsf dx = \left[\log \left(\tan\left(\tfrac{x}{2}\right)\right) - \log (x)\right]\bigg|_{0}^{\pi/2} = \log \left(\frac{2}{\pi}\right) - \lim\limits_{x \to 0^+} \log \left(\frac{\tan(x/2)}{x}\right) = \log \left(\frac{4}{\pi}\right).
\end{equation*}
Note that the limit is taken from above. Thus,
\begin{equation*}
    G(2n) = \frac{2H_{n - 1}}{\pi} + \frac{2}{\pi} \log \left(\frac{4}{\pi}\right) + o(1) = \frac{2\log (n) + 2\gamma}{\pi} + \frac{2}{\pi} \log \left(\frac{4}{\pi}\right) + o(1).
\end{equation*}

The substitution of $n/2$ instead of $n$ concludes the proof of Lemma~\ref{lem:boundonG(n)}.

\section{Basic Preliminaries}
We consider Gaussian differential privacy as the privacy notion in this paper. Informally speaking, a mechanism is $\mu$-Gaussian differential privacy if for neighboring datasets $D$ and $D'$ and given real-valued function $f$, distinguishing $f(D)$ and $f(D')$ is no easier than distinguishing between the Gaussian distributions $\mathcal N(0,1)$ and $\mathcal N(\mu,1)$. 
\begin{definition}
    [Gaussian Differential Privacy]
\label{defn:DP}
A randomized algorithm $\mathcal M$  is $\mu$-Gaussian differentially private ($\mu$-GDP) if, for any neighboring datasets $D$ and $D'$ (denoted $D \simeq D'$), there is a (possibly randomized) function $g :\mathbb R \to Y$ with
\begin{align*}
    g(z) \stackrel{d}{=} \mathcal M(D) \text{ for }  z \sim \mathcal N(0,1), \quad \text{and} \quad 
    g(z) \stackrel{d}{=} \mathcal M(D') \text{ for }  z \sim \mathcal N(\mu,1),
\end{align*}
where $X \stackrel{d}{=} Y$ means random variables $X$ and $Y$ are identically distributed.
\end{definition}

One of the canonical mechanisms to preserve $\mu$-GDP is the Gaussian mechanism. To define the Gaussian mechanism, we need to define $\ell_2$-sensitivity: 
\[
\Delta_2(f):= \max_{D \simeq D'} \norm{f(D)-f(D')}_2.
\]

\begin{definition}
    [Gaussian mechanism]
    \label{def:gaussian_mech}
    Given a function $f: \mathcal X^* \to \mathbb R^m$ with $\ell_2$-sensitivity $\Delta_2(f)$ and scale multiplier $1/\mu$, set  $\sigma=\tfrac{1}{\mu}\Delta_2(f)$. Then, the Gaussian mechanism on input $D$, outputs $y \gets f(D)+z$, where $z\sim \mathcal N(0,\sigma^2\mathbb I_m)$ is a multivariate Gaussian with every entry picked i.i.d. from $\mathcal N(0,\sigma^2)$. Furthermore, $y$ is $\mu$-GDP.
\end{definition}

\section{Our Factorization and Proof of the Upper Bounds in Theorem~\ref{thm:mainTheorem}}
\label{sec:our_factorization}

For the prefix sum matrix $A$, the square root matrix $C$ is a Lower Triangular Toeplitz (LTT) matrix defined by coefficients $r_j$ on the diagonals for $j \in [0, n - 1]$, where $r_n$ are the coefficients of $\frac{1}{\sqrt{1-x}}$. The following two lemmas tightly bound these coefficients:
\begin{lemma}[Wallis inequality~\cite{chen2005best}]
\label{lem:WallisInequality}
For $n \in \mathbb N$, we have      
    \begin{equation*}
     \frac{1}{\pi (n + 1)}\le \frac{1}{\pi(n + \tfrac{4}{\pi} - 1)} \leq r_n^2 \leq \frac{1}{\pi\paren{n+\frac{1}{4}}} \le \frac{1}{\pi n}, \quad \text{where} \quad r_n :=\begin{cases}
        1 & n = 0 \\
        \frac{1}{4^n}\binom{2n}{n} & n \geq 1
    \end{cases}.
    \end{equation*}
\end{lemma}

\begin{lemma}
\label{lem:monotonicityofa_n}
    For $n \in \mathbb N_+$ in the set of positive integers,  define the sequence  
    \begin{align}
    \label{eq:alpha_n}        
    \alpha_n = \sum\limits_{j = 0}^{n - 1} r_j^2 - \frac{\log (n)}{\pi}, \quad \text{where} \quad 
    r_k :=\begin{cases}
        1 & k = 0 \\
        \frac{1}{4^k}\binom{2k}{k} & k \geq 1
    \end{cases}
    \end{align}
    Then $\alpha_n$ is monotonically increasing in $n \in \mathbb N$. Further, the sequence satisfies $1 \leq \alpha_n \leq 1.0663$.
\end{lemma}
%\nk{We need $\epsilon_n$ as well, for the MeanSE bound.}
\begin{proof}
    
Note that 
\begin{equation*}
    \alpha_{n+1} - \alpha_n = r_n^2 - \frac{\log(n+1) - \log(n)}{\pi} = r_n^2 - \frac{\log(1 + 1/n)}{\pi}.
\end{equation*}

We aim to prove that this difference is positive. Unfortunately, the standard bound $r_n^2 \ge \frac{1}{\pi(n + 1)}$ is not tight enough to establish this\footnote{We also note that the more classical bound $r_n \ge \frac{1}{\sqrt{\pi\left(n + {1}/{2}\right)}}$ doesn't suffice either, making this an unprecedented case in matrix factorization analysis.}. Instead, we refer to the best known lower bound on Wallis' inequality due to Chen and Qi~\cite{chen2005best}:
\begin{equation*}
    r_n^2 \ge \frac{1}{\pi(n + \tfrac{4}{\pi} - 1)}.
\end{equation*}

Combining this with the general Taylor-series-based upper bound for the logarithm, $\log(1 + x) \le x - \tfrac{x^2}{2} + \tfrac{x^3}{3}$, we obtain
\begin{equation*}
    \alpha_{n+1} - \alpha_n \ge \frac{1}{\pi(n + \tfrac{4}{\pi} - 1)} - \frac{1}{\pi n} + \frac{1}{2\pi n^2} - \frac{1}{3\pi n^3} = \frac{9\pi n^2 - 24n^2 - 5\pi n + 12n + 2\pi - 8}{6\pi n^3(\pi n - \pi + 4)},
\end{equation*}
which is greater than zero when
\begin{equation*}
    n \ge \frac{5\pi - 12}{6(3\pi - 8)} + \frac{\sqrt{-624 + 360\pi - 47\pi^2}}{6(3\pi - 8)} \approx 1.2017.
\end{equation*}

We verify the base case $n = 1$ separately: $\alpha_1 = 1$, $\alpha_2 = 1 + \tfrac{1}{4} - \tfrac{\log (2)}{\pi} \approx 1.03 > 1$, concluding the proof of monotonicity. Lemma~\ref{lem:monotonicityofa_n} follows by noting that $\alpha_n$ appears in Landau's constant with an explicit bound given in~\cite{brijtman1982sharp}. 
\end{proof}

The inverse of this the square root of $\counting$, $C^{-1}$, has been studied  by Kalinin et al.~\cite[Lemma 3]{kalinin2025back}, and its entries take the following form: 
\begin{equation*}
    C^{-1} = \begin{pmatrix}
        \widetilde{r}_0 & 0 & \dots & 0\\
        \widetilde{r}_1 & \widetilde{r}_0 & \dots & 0\\
        \vdots & \vdots & \ddots & \vdots\\
        \widetilde{r}_{n - 1} & \widetilde{r}_{n - 2} & \dots & \widetilde{r}_0
    \end{pmatrix}, 
    \quad \text{where} \quad 
    \widetilde r_j =\begin{cases}
        1 & k = 0 \\
        \frac{-1}{2j - 1} r_j  & j \geq 1 
    \end{cases}
\end{equation*}
are the coefficients in the Taylor series of $\sqrt{1-x}$.

From the definition, the identity $A C^{-1} = C$ follows, which implies a prefix sum identity for the values of $\widetilde{r}_j$:
\begin{equation}
    \sum\limits_{t = 0}^{j} \widetilde{r}_t = r_j.
    \label{eq:summationr_n}
\end{equation}

The matrix $\widetilde{B}$ can be computed explicitly. Let $d_j$ be the norm of the $j$-th column of the matrix $C$, given by the positive $d_j \geq 0$ (for $j\geq 0$)  such that  %\ju{We should have a lemma stating this.}
\begin{equation}
\label{eq:d_j_definition}
    d_j^2 = {\sum\limits_{t = 0}^{n - j} r_t^2} = {\alpha_{n - j + 1} + \frac{\log (n - j + 1)}{\pi}} \le {\alpha_{\infty} + \frac{\log (n - j + 1)}{\pi}}\qquad \text{for } j \in [1, n],
\end{equation}
where $\alpha_{n - j + 1}$ is defined in Theorem~\ref{thm:square_root_bounds} and bounded from below by $1$ and upper bounded by $\alpha_{\infty} \approx 1.0663$.

Then, for $j \ge k$, the entry $\widetilde{B}_{j,k}$ can be computed as:
\begin{equation}
\label{eq:b_j_k_expression}
    \widetilde{B}_{j,k} = \sum\limits_{t = 0}^{j - k} \widetilde{r}_{t} d_{k + t} \qquad \text{for } j, k \in [1, n].
\end{equation}

To compute the MaxSE, we break the proof in several steps:
\begin{enumerate}
    \item In Lemma~\ref{lem:B_j_k_upper_bound}, we compute an upper and lower bound on the entries of the matrix, $\widetilde B$, i.e.,  $\widetilde{B}_{j,k}$ for any $1 \leq j,k \leq n$.
    \item We use Lemma~\ref{lem:B_j_k_upper_bound} to bound the squared $\ell_2$-norm of any row of the matrix $\widetilde{B}$ in Lemma~\ref{lem:b_j_row_norm_bound}, i.e., $\norm{B_{j,:}}^2$. 
    \item Using Lemma~\ref{lem:b_j_row_norm_bound}, we bound $\norm{B}_{2 \to \infty}$ in Lemma~\ref{lem:rowNormOfB}. 
\end{enumerate}

%In the next lemma, we prove upper and lower bounds for the entries of the matrix $\widetilde{B}$.

We will extensively use the following indefinite integral in the proof, 
evaluated for different values of $a \in \left\{ -\tfrac{1}{2}, \, 0, \, \tfrac{1}{2}, \, 1 \right\}$: 
\begin{equation} \label{eq:log_integral}
    \int \frac{\log^a(x)}{x}\, \mathrm{d}x 
    = \frac{1}{a+1}\,\log^{\,a+1}(x) + \mathrm{const}.
\end{equation}

\begin{lemma}
[Bound on entries of $\widetilde B$]
\label{lem:B_j_k_upper_bound}
For $j > k$ the following upper bound for $\widetilde{B}_{j,k}$ holds true:
\begin{equation*}
   d_j r_{j - k} \le \widetilde{B}_{j, k}  \le d_j r_{j - k} +  \frac{\sqrt{3\log(n-k)} - \sqrt{3\log(n - j + 1)}}{\pi\sqrt{2(n - k)}}  + \frac{1}{\sqrt{n - k}},
\end{equation*}
and for $j = k$ we have $\widetilde{B}_{j,j} = d_j r_{0} = d_j$. 
\end{lemma}
\begin{proof}

We start from the formula for $\widetilde{B}_{j, k}$, equation \eqref{eq:b_j_k_expression},  and apply summation by parts (also known as the Newton series or Abel transformation):
 \begin{align*}
        \widetilde{B}_{j,k} &= \sum\limits_{t = 0}^{j - k} \widetilde{r}_{t } d_{k + t} = d_j {\sum\limits_{t = 0}^{j - k}\widetilde{r_t}} - \sum\limits_{t = 0}^{j - k - 1} (d_{k + t + 1} - d_{k + t})\sum\limits_{l = 0}^{t} \widetilde{r}_l \tag{using eq. \eqref{eq:b_j_k_expression}} \\
        &= d_j r_{j - k} + \sum\limits_{t = 0}^{j - k - 1} (d_{t + k} - d_{t + k + 1})r_t. \tag{using eq. \eqref{eq:summationr_n}}
    \end{align*}
This implies the bound for $\tilde{B}_{j,j}$. 
Now we recall the definition of $d_t$ in eq. \eqref{eq:d_j_definition}. The sequence is decreasing, which immediately gives us the desired lower bound. 

We can rewrite the difference $d_{t + k} - d_{t + k + 1}$ by using the fact that $r^2_{n-(t+k)}=d_{t+k}^2 - d^2_{t+k+1}$ from the definition of the sequence $d_j^2$ defined in eq. \eqref{eq:d_j_definition}. Further, $d_j$ is an increasing sequence. Therefore, 
\begin{equation}
\label{eq:B_j_k_first_inequality}
\widetilde{B}_{j,k} = d_j r_{j - k} + \sum\limits_{t = 0}^{j - k - 1} \frac{r_{n - k - t}^2}{d_{t + k} + d_{t + k + 1}} r_t \le d_j r_{j - k} + \frac{1}{2} \sum\limits_{t = 0}^{j - k - 1} \frac{r_{n - k - t}^2r_t}{d_{t + k + 1}},
\end{equation}

For the rest of the sum we first consider the cases $t = 0$ and $t = j - k - 1$ separately:

\begin{equation*}
    \widetilde{B}_{j, k} \le d_j r_{j - k} + \frac{r^2_{n - k} }{2d_{k + 1}} + \frac{r_{j - k - 1}r_{n - j + 1}^2}{2d_j} + \frac{1}{2}\sum\limits_{t = 1}^{j - k - 2} \frac{r_{n - k - t}^2r_t}{d_{t + k + 1}}.
\end{equation*}

Note that, from eq. \eqref{eq:d_j_definition}, we have 
\[
d_{t + k + 1} \ge \max \left\{\sqrt{\tfrac{\max\{1, \log (n - k - t)\} }{\pi}},1 \right\}
\]

Using the Wallis inequality (Lemma~\ref{lem:WallisInequality}), this implies that 
\begin{align}
    \frac{r_{n - k - t}^2 r_t}{d_{t + k + 1}} \leq  \frac{1}{\pi(n - k - t)\sqrt{t \cdot \max(\log(n - k - t), 1)}}.
    \label{eq:rsquaredrdBound}
\end{align}
%Using $\alpha_{n - k - t} \ge 0$, this implies that the  remainder of the sum can be bounded as
We now split the sum into two parts: values $t \le \frac{2(n - k)}{3}$, and values larger than $\frac{2(n - k)}{3}$, for which we will apply different inequalities. In particular, we decompose the sum as  
\begin{equation*}
    \frac{1}{2} \sum\limits_{t = 1}^{j - k - 2} \frac{r_{n - k - t}^2 r_t}{d_{t + k + 1}} 
    = \underbrace{\frac{1}{2} \sum\limits_{t = 1}^{ \lfloor \frac{2(n - k)}{3} \rfloor} \frac{r_{n - k - t}^2 r_t}{d_{t + k + 1}}}_{S_1} + \underbrace{\frac{1}{2} \sum\limits_{t =  \lfloor \frac{2(n - k)}{3} \rfloor + 1}^{j - k - 2} \frac{r_{n - k - t}^2 r_t}{d_{t + k + 1}}}_{S_2}.
    %\le \frac{1}{2\pi} \sum\limits_{t = 1}^{j - k - 2} \frac{1}{(n - k - t)\sqrt{t \cdot \max(\log(n - k - t), 1)}}.
\end{equation*}

We begin with the latter case, assuming $j - k - 2 \ge \lfloor \frac{2(n - k)}{3} \rfloor + 1$. If not, then this sum is identically $0$. 
First using eq. \eqref{eq:rsquaredrdBound}, we have 
\begin{align*}
    S_2 &= \frac{1}{2} \sum\limits_{t =  \left\lfloor \frac{2(n - k)}{3} \right\rfloor +  1}^{j - k - 2} \frac{r_{n - k - t}^2 r_t}{d_{t + k + 1}} \leq \frac{1}{2\pi}\sum\limits_{t = \left\lfloor \frac{2(n - k)}{3} \right\rfloor + 1}^{j - k - 2} \frac{1}{(n - k - t)\sqrt{t\cdot \max(\log(n - k - t), 1)}}  \\
    &\le \frac{1}{2\pi\sqrt{\frac{2(n - k)}{3}}}\sum\limits_{t = \left\lfloor \frac{2(n - k)}{3} \right\rfloor + 1}^{j - k - 2} \frac{1}{(n - k - t)\sqrt{\log(n - k - t)}}.
\end{align*}
Changing the variable gives us 
\begin{align*}
    S_2 
    &\le  \frac{1}{2 \pi}\sqrt{ \frac{3}{2(n-k)}} \sum\limits_{t= n -j + 2}^{\lceil \frac{n - k}{3} \rceil} \frac{1}{t\sqrt{\log t}}\le \frac{1}{2 \pi}\sqrt{\frac{3}{2(n-k)}} \int\limits_{t= n -j + 1}^{\lceil \frac{n - k}{3} \rceil} \frac{\mathsf dt}{t\sqrt{\log t}}\\
    &\le \frac{1}{\pi}\sqrt{\frac{3}{2(n-k)}} \left[\sqrt{\log\left(\left\lceil \frac{n - k}{3} \right\rceil\right)} - \sqrt{\log\left(n -j + 1\right)}\right] \\
    &= \frac{1}{\pi}\paren{\sqrt{\frac{3\log(n-k)}{2(n-k)}} - \sqrt{\frac{3\log(n-j+1)}{2(n-k)}}}
%    & \le \frac{\sqrt{3/2}(\sqrt{\log(n-k)} - \sqrt{\log(n - j + 1)})}{\pi\sqrt{n - k}} 
\end{align*}

For the rest of the sum we use the following bound:
\begin{equation*}
\frac{1}{2\pi}\sum\limits_{t = 1}^{\min\left\{\left\lfloor \frac{2(n - k)}{3} \right\rfloor,\; j - k - 2\right\}} \frac{1}{(n - k - t)\sqrt{t\cdot\max(\log(n - t- k), 1)}} \le \frac{1}{2\pi} \sum\limits_{t = 1}^{\lfloor \frac{2(n - k)}{3} \rfloor} \frac{1}{(n - k - t)\sqrt{t}}.
\end{equation*}

To bound the new sum, we observe that the summand is not monotonic in $t$. However, we can identify regions where it is increasing or decreasing. Specifically, consider:

\begin{equation*}
\frac{\mathsf d}{\mathsf dx}\left(\frac{1}{(a - x)\sqrt{x}}\right) = \frac{3x - a}{2(a - x)^2 x^{3/2}}.
\end{equation*}

Therefore, for $t \le \frac{n - k}{3}$ the function is decreasing, and for $t \ge \frac{n - k}{3}$ it is increasing. This allows us to apply integral bounds. In the first step, we remove the border case to get
\begin{align}
\sum\limits_{t = 1}^{\left\lfloor \frac{2(n - k)}{3} \right\rfloor} \frac{1}{(n - k - t)\sqrt{t}} 
\le \frac{3}{2(n - k)} +  \sum\limits_{t = 1}^{\left\lfloor \frac{2(n - k)}{3} \right\rfloor - 1} \frac{1}{(n - k - t)\sqrt{t}}.
\label{eq:summationFirstCase}
\end{align}

We now break the remaining sum into increasing and decreasing function parts so that we can apply the integral bound appropriately. That is, 

\begin{align*}
\sum\limits_{t = 1}^{\lfloor \frac{2(n - k)}{3} \rfloor-1} \frac{1}{(n - k - t)\sqrt{t}} 
& \le \int\limits_{0}^{\lfloor \frac{n - k}{3} \rfloor} \frac{1}{(n - k - x)\sqrt{x}}\, \mathsf dx + \int\limits_{\lfloor \frac{n - k}{3} \rfloor + 1}^{\lfloor \frac{2(n - k)}{3} \rfloor} \frac{1}{(n - k - x)\sqrt{x}}\, \mathsf dx \\
&\le \int\limits_{0}^{\lfloor \frac{2(n - k)}{3} \rfloor} \frac{1}{(n - k - x)\sqrt{x}}\, \mathsf dx .
\end{align*}

 The integral on the right-hand side can be computed explicitly:

\begin{equation*}
\int \frac{1}{(a - x)\sqrt{x}}  \mathsf dx = \frac{2\tanh^{-1}(\sqrt{x/a})}{\sqrt{a}} + \mathrm{const}.
\end{equation*}

As $\tanh^{-1}\left(\sqrt{\frac{\lfloor \frac{2(n - k)}{3} \rfloor }{n - k}}\right) \leq \frac{1}{2}\log(5)< 1$, taking the limit and plugging it back to eq~\eqref{eq:summationFirstCase}, we get 

\begin{align*}
\frac{1}{2\pi}\sum\limits_{t = 1}^{\lfloor \frac{2(n - k)}{3} \rfloor} \frac{1}{\sqrt{t}(n - k - t)\sqrt{\log(n - t- k)}} &\le  \frac{1}{2\pi \sqrt{n - k}} \tanh^{-1}\left(\sqrt{\frac{\lfloor \frac{2(n - k)}{3} \rfloor }{n - k}}\right) + \frac{3}{{4\pi}(n - k)}\\
&\le \frac{3}{{4\pi}(n - k)} + \frac{1}{2\pi\sqrt{n - k}}  \le \frac{{5}}{4\pi\sqrt{n - k}}.
\end{align*}

We now bound the two extra terms that we separated earlier, using that $d_j \ge 1$ for all $j$:
\begin{equation*}
    \frac{r^2_{n - k} }{2d_{k + 1}} \le \frac{1}{2\pi (n - k)} \le \frac{1}{2\pi\sqrt{n - k}} \text{ and }
    \frac{r_{j - k - 1}r_{n - j + 1}^2}{2d_j} \le \frac{1}{2\pi^{3/2}\sqrt{j - k - 1}(n - j + 1)} \le \frac{1}{\sqrt{2}\pi^{3/2} \sqrt{n - k}},
\end{equation*}
where the last inequality follows from the fact that $\sqrt{2ab^2} = \sqrt{a(b^2) + b(ab)} \ge \sqrt{a + b}$,  with $a = j - k - 1 \ge 1$ and $b = n - j + 1 \ge 1$, and consequently $\sqrt{2(j - k-1)(n - j + 1)^2}\ge \sqrt{a + b} = \sqrt{n - k}$. We now combine all terms to obtain:

\begin{align*}
   \widetilde{B}_{j, k} &\le d_j r_{j - k} + \frac{\sqrt{3}(\sqrt{\log(n-k)} - \sqrt{\log(n - j + 1)})}{\pi\sqrt{2(n - k)}}  + \frac{1}{\sqrt{n - k}}\left({\frac{5}{4\pi}}  + \frac{1}{2\pi} + \frac{1}{\sqrt{2}\pi^{3/2}}\right)\\
   &\le d_j r_{j - k} + \frac{\sqrt{3}(\sqrt{\log(n-k)} - \sqrt{\log(n - j + 1)})}{\pi\sqrt{2(n - k)}}  + \frac{{1}}{\sqrt{n - k}},
\end{align*}
which concludes the proof of Lemma~\ref{lem:B_j_k_upper_bound}.
\end{proof}

\begin{lemma}
[Bound on squared $\ell_2$ norm of rows of $\widetilde B$]
\label{lem:b_j_row_norm_bound}
The squared $j$-th row norm of matrix $\widetilde{B}$ has the following upper and lower bound:
    \begin{align*}
        d_j^2 d_{n - j + 1}^2\le \sum\limits_{k = 1}^{j} (\widetilde{B}_{j, k})^2 &\le d_j^2d_{n - j + 1}^2 +  \frac{3}{4\pi^2} \log^2 \left(\frac{n - 1}{n - j + 1}\right) + \frac{\sqrt{6}d_j}{\pi^{3/2}\sqrt{\log (n - 1})}\log^2\left( \frac{4n}{n - j + 1} \right)\\
        &\qquad + 3\sqrt{\log (n)}\log\left(\frac{n}{n - j + 1}\right) + o(\log (n))
    \end{align*}
\end{lemma}
\begin{proof}
We use Lemma~\ref{lem:B_j_k_upper_bound} to bound the values of $\widetilde{B}_{j, k}$. Since we are squaring these bounds, we need to compute six corresponding sums, keeping track of terms up to order $o(\log (n))$. This level of precision is sufficient, as we only require an $o(1)$ error for MaxSE and MeanSE after taking the square root. For $j = k$, we have $\widetilde{B}_{j, j} = d_j$; for $j > k$:

\begin{equation*}
0 < d_j r_{j - k} \le \widetilde{B}_{j, k} \le d_j r_{j - k} + \frac{\sqrt{3\log(n - k)} - \sqrt{3\log(n - j + 1)}}{\pi \sqrt{2(n - k)}} + \frac{1}{\sqrt{n - k}}
\end{equation*}

Note that the second term is equal to $0$ for $k = j - 1$ and does not appear when $j = k$, so the second term only needs to be summed up to $j - 2$, requiring $j \ge 3$. 

\begin{align*}
\underbrace{\sum_{k=1}^jd_j^2 r_{j-k}^2}_{S_{(1)}} \le \sum_{k=1}^j \widetilde{B}_{j,k}^2 \le\ 
&  \underbrace{\sum_{k=1}^jd_j^2 r_{j-k}^2}_{S_{(1)}} 
+ \underbrace{\sum_{k=1}^{j - 2} \frac{3(\sqrt{\log(n-k)} - \sqrt{\log(n - j + 1)})^2}{2\pi^2(n - k)} }_{S_{(2)}} \nonumber \\
& + \underbrace{\sum_{k=1}^{j - 2}  d_j r_{j-k} \cdot \frac{(\sqrt{6\log(n-k)} - \sqrt{6\log(n - j + 1)})}{\pi\sqrt{n - k}}}_{S_{(3)}}  
+ \underbrace{\sum_{k=1}^{j - 1} d_j r_{j-k} \cdot \frac{2}{\sqrt{n-k}}}_{S_{(4)}} \nonumber \\
& +  \underbrace{\sum_{k=1}^{j - 2}\frac{(\sqrt{6\log(n-k)} - \sqrt{6\log(n - j + 1)})}{\pi(n - k)}}_{S_{(5)}} \nonumber  +  \underbrace{\sum_{k=1}^{j - 1}\frac{1}{n-k}}_{S_{(6)}}.
\end{align*}

Now we bound each of these terms separately. 
\subsection*{\texorpdfstring{Bounding $S_{(1)}$}{Bounding S(1)}}
The lower bound follows from the definition of $d_j^2$ (eq. \eqref{eq:d_j_definition}) and the fact that rest of all the terms are positive.   
\begin{equation*}
    \sum\limits_{k = 1}^{j} d_j^2 r_{j - k}^2 = d_j^2 \sum\limits_{k = 0}^{j - 1} r_{k}^2 = d_{j}^2 d_{n - j + 1}^2,
\end{equation*}
%which gives the desired lower bound.

\subsection*{\texorpdfstring{Bounding $S_{(2)}$}{Bounding S(2)}}
Note that, for $0\leq b \leq a$,  $a^2-b^2=(a-b)(a+b) \geq (a-b)^2$. 
Setting $a=\sqrt{\log(n-k)}$ and $b=\sqrt{\log(n-j+1)}$ thus gives us  
\begin{align*}
    S_{(2)} &= \frac{3}{2\pi^2}\sum\limits_{k = 1}^{j - 2}\frac{\left(\sqrt{\log(n - k)} - \sqrt{\log(n - j + 1)}\right)^2}{n - k} \le \frac{3}{2\pi^2}\sum\limits_{k = 1}^{j - 2}\frac{\log(n - k) - \log(n - j + 1)}{n - k} \\
    &\le \frac{3}{2\pi^2} \sum\limits_{k = n - j + 2}^{n - 1} \frac{\log k}{k} - \frac{3\log(n - j + 1)}{2\pi^2}\sum\limits_{k = n - j + 2}^{n - 1} \frac{1}{k}
\end{align*}

We use integral bounds for the sums. In the first sum, $k$ should be at least $3$ so that the terms are monotonically decreasing. We need to treat $k = n - j + 2$ separately; however, we argue that its contribution is $o(\log (n))$ and will disregard it. Similarly, we approximate the second sum using an integral from $n - j + 1$ to $n - 1$, ignoring boundary terms since they also contribute $o(\log (n))$. This results in the following bound:
\begin{align*}
    S_{(2)} & \leq \frac{3}{2\pi^2} \left[\frac{1}{2}(\log^2(n - 1) - \log^2(n - j + 1)) - \log(n - j + 1)\log\left(\frac{n - 1}{n - j + 1}\right)\right] + o(\log (n))\\
    & = \frac{3}{4\pi^2}\left[\log^2(n - 1) -2\log(n - j + 1)\log(n - 1)+ \log^2(n - j + 1)\right] + o(\log (n))\\
    & = \frac{3}{4\pi^2} \log^2\left(\frac{n - 1}{n - j + 1 }\right) + o(\log (n)). 
\end{align*}

\subsection*{\texorpdfstring{Bounding $S_{(3)}$}{Bounding S(3)}}
Using the identity $a^2-b^2=(a-b)(a+b)$, we get 

\begin{align*}
    S_{(3)} & = \frac{\sqrt{6}d_j}{\pi}\sum\limits_{k = 1}^{j - 2} \frac{r_{j - k}(\sqrt{\log(n -k )} - \sqrt{\log(n - j + 1)})}{\sqrt{n - k}} \\
    &\le \frac{\sqrt{6}d_j}{\pi^{3/2}}\sum\limits_{k = 1}^{j - 2} \frac{\sqrt{\log(n -1)} - \sqrt{\log(n - j + 1)}}{\sqrt{j - k}\sqrt{n - k}}\\
    &=\frac{\sqrt{6}d_j\log \left(\frac{n - 1}{n - j + 1}\right)}{\pi^{3/2}(\sqrt{\log (n - 1)} + \sqrt{\log (n - j+1)})}\sum\limits_{k = 1}^{j - 2} \frac{1}{\sqrt{j - k}\sqrt{n - k}}
\end{align*}
Since $\sqrt{\log(n-j+1)}\geq 0$, we can thus write 
\begin{align}
    S_{(3)} \le \frac{\sqrt{6}d_j}{\pi^{3/2}\sqrt{\log (n - 1)}}\log\left(\frac{n - 1}{n - j + 1}\right)\sum\limits_{k = 1}^{j - 2} \frac{1}{\sqrt{j - k}\sqrt{n - k}}
\label{eq:boundonS3}
\end{align}

Now we bound the remaining sum:
\begin{equation}
\begin{aligned}
\label{eq:double_sum}
    \sum_{k = 1}^{j - 2} \frac{1}{\sqrt{j - k}\sqrt{n - k}}
    &\le \sum_{k = 1}^{j - 2} \frac{1}{\sqrt{j - 1 - k}\sqrt{n - k}} \le \sum_{k = 1}^{j - 2} \frac{1}{\sqrt{k}\sqrt{n - j + 1 + k} } \\
    &\le \int\limits_{0}^{j - 1} \frac{\mathsf dx}{\sqrt{x}\sqrt{n - j + 1 + x} } = 2\tanh^{-1}\left( \sqrt{\frac{j - 1}{n}} \right) \\
    &= \log\left(\frac{\sqrt{n} + \sqrt{j - 1}}{\sqrt{n} - \sqrt{j - 1}}\right) = \log\left(\frac{(\sqrt{n} + \sqrt{j - 1})^2}{n - j + 1}\right) \le \log\left( \frac{4n}{n - j + 1} \right)
\end{aligned}
\end{equation}

Plugging eq. \eqref{eq:double_sum} in eq. \eqref{eq:boundonS3}, we get 

\begin{equation*}
    S_{(3)} \leq \frac{\sqrt{6}d_j}{\pi^{3/2}\sqrt{\log (n - 1})}\log\left(\frac{n - 1}{n - j + 1}\right)\log\left( \frac{4n}{n - j + 1} \right) \le \frac{\sqrt{6}d_j}{\pi^{3/2}\sqrt{\log (n - 1})}\log^2\left( \frac{4n}{n - j + 1} \right)
\end{equation*}

\subsection*{\texorpdfstring{Bounding $S_{(4)}$}{Bounding S(4)}}
Since $d_j \le \sqrt{\alpha_{\infty} + \frac{\log(n - j + 1)}{\pi}}$ (see equation \eqref{eq:d_j_definition}). Therefore, $\frac{d_j}{\sqrt{n - j + 1}}$ is bounded by a constant. That is, 
\begin{align*}
S_{(4)} 
    & = \sum_{k=1}^{j - 1} d_j r_{j-k} \cdot \frac{2}{\sqrt{n-k}} \le \frac{2d_j}{\sqrt{\pi}} \sum\limits_{k =1}^{j - 1} \frac{1}{\sqrt{j - k} \sqrt{n - k}} \\
    &= \frac{2d_j}{\sqrt{\pi}} \sum\limits_{k =1}^{j - 2} \frac{1}{\sqrt{j - k} \sqrt{n - k}} + o(\log (n)) \le \frac{2d_j}{\sqrt{\pi}}\log\left( \frac{4n}{n - j + 1} \right) + o(\log (n)). \tag{eq. \eqref{eq:double_sum}}
\end{align*}

\subsection*{\texorpdfstring{Bounding $S_{(5)}$}{Bounding S(5)}}

\begin{align*}
    S_{(5)} &= \sum_{k=1}^{j - 2}\frac{\sqrt{6}(\sqrt{\log(n-k)} - \sqrt{\log(n - j + 1)})}{\pi(n - k)} = \frac{\sqrt{6}}{\pi} \sum\limits_{k = n - j + 2}^{n - 1} \frac{\sqrt{\log k}}{k} - \frac{\sqrt{6\log(n - j + 1)}}{\pi}\sum\limits_{j = n - j + 2}^{n - 1}\frac{1}{k}\\
    &=\frac{\sqrt{6}}{\pi} \sum\limits_{k = n - j + 2}^{n - 1} \frac{\sqrt{\log k}}{k} - \frac{\sqrt{6\log(n - j + 1)}}{\pi}\sum\limits_{j = n - j + 1}^{n - 2}\frac{1}{k} + o(\log (n))
\end{align*}

We use integral bounds, noting that for $k \ge 2$ the function $\frac{\sqrt{\log k}}{k}$ is decreasing. This results in: 
\begin{align*}
    S_{(5)}&\leq \frac{\sqrt{6}}{\pi} \int\limits_{n - j + 1}^{n - 1} \frac{\sqrt{\log (x)}}{x}\, \mathsf dx - \frac{\sqrt{6\log(n - j + 1)}}{\pi} \int\limits_{n -j + 1}^{n - 1} \frac{\mathsf dx}{x} + o(\log (n))\\
    &=\frac{\sqrt{6}}{\pi} \left(\frac{2}{3}\log^{3/2}(n - 1) - \frac{2}{3}\log^{3/2}(n - j + 1)\right)\\
    &\qquad - \frac{\sqrt{6}}{\pi}\left(\log^{1/2}(n - j + 1)\log (n - 1) - \log^{3/2}(n - j + 1)\right) + o(\log (n))\\
    & = \frac{\sqrt{6}}{\pi} \frac{(\sqrt{\log (n - 1)} - \sqrt{\log (n - j + 1)})^2}{3} \left(2 \sqrt{\log (n - 1)} + \sqrt{\log (n - j + 1)}\right) + o(\log (n))\\
    & \le \frac{\sqrt{6}\sqrt{\log(n - 1)}}{\pi} \log\left(\frac{n - 1}{n - j + 1}\right) + o(\log (n)).
\end{align*}

\subsection*{\texorpdfstring{Bounding $S_{(6)}$}{Bounding S(6)}}

Again using the integral bound, we have 
\begin{align*}
    S_{(6)} &= \sum\limits_{k = 1}^{j - 1} \frac{1}{n - k} \le \frac{1}{n - j + 1} + \sum\limits_{k = n - j + 2}^{n - 1} \frac{1}{k} \le  \int\limits_{n - j + 1}^{n - 1} \frac{\mathsf dx}{x} + o(\log (n)) \\
    &= \log\left(\frac{n - 1}{n - j + 1}\right) + o(\log (n)) 
\end{align*}

\subsection*{Combining all together}  We note that the first three terms in Lemma~\ref{lem:b_j_row_norm_bound} are the same as the bound we have for $S_{(1)},$ $S_{(2)},$ and $S_{(3)}.$ Therefore, it suffices to bound the last three terms. We do it as follows: First note that, 
\begin{equation*}
d_j = \sqrt{\tfrac{\log(n - j + 1)}{\pi} + \alpha_{n - j + 1}} < \sqrt{\tfrac{\log(n)}{\pi} + \alpha_{n - j + 1}} = \sqrt{\tfrac{\log (n - 1)}{\pi}} + o(1),
\end{equation*}
Therefore,

\begin{align*}
    S_{(4)} &+ S_{(5)} + S_{(6)} \leq \frac{2d_j}{\sqrt{\pi}}\log\left( \frac{4n}{n - j + 1} \right) + \frac{\sqrt{6\log(n - 1)}}{\pi} \log\left(\frac{n - 1}{n - j + 1}\right) + \log\left(\frac{n - 1}{n - j + 1}\right)  + o(\log (n))\\
    &\le \frac{2\sqrt{\log(n - 1)}}{\pi}\log\left( \frac{n - 1}{n - j + 1}  \right) + o(\log (n)) + \left(\frac{\sqrt{6}}{\pi} + 1\right)\sqrt{\log(n - 1)} \log\left(\frac{n - 1}{n - j + 1}\right) + o(\log (n))\\
    &=  \left(\frac{2}{\pi} +\frac{\sqrt{6}}{\pi} + 1\right)\sqrt{\log(n - 1)} \log\left(\frac{n - 1}{n - j + 1}\right) + o(\log (n))\\
    &\le 3\sqrt{\log(n - 1)}\log\left(\frac{n - 1}{n - j + 1}\right) + o(\log (n)) = 3\sqrt{\log (n)}\log\left(\frac{n}{n - j + 1}\right) + o(\log (n)),
\end{align*}
concluding the proof of the Lemma~\ref{lem:b_j_row_norm_bound}.
\end{proof}

\begin{lemma}
[Bound on row-norm of $\widetilde B$]
\label{lem:rowNormOfB}
Let $(\alpha_m)_{n\geq 0}$ be the sequence defined in eq. \eqref{eq:alpha_n}. 
Then the maximum squared row norm of the matrix $\widetilde{B}$ has the following expression:
\begin{equation*}
   \norm{\widetilde{B}}_{2\to \infty} = \max_{j}  \sum\limits_{k = 1}^{j} (\widetilde{B}_{j, k})^2 = \frac{\log^2 (n) }{\pi^2} + \frac{2(\alpha_{\infty} - \frac{1}{\pi}\log(2))\log(n)}{\pi}  + o(\log (n)) \quad \text{where} \quad \alpha_\infty := \lim_{m \to \infty} \alpha_m.
\end{equation*}
\end{lemma}

\begin{proof}
We recall the bounds for the $j$-th row of the matrix $\widetilde{B}$ from Lemma~\ref{lem:b_j_row_norm_bound}:
\begin{align*}
   d_j^2 d_{n - j + 1}^2 \le  \sum\limits_{k = 1}^{j} (\widetilde{B}_{j, k})^2 \le d_j^2 d_{n - j + 1}^2 &+ \frac{3}{4\pi^2} \log^2 \left(\frac{n - 1}{n - j + 1}\right) + \frac{\sqrt{6} d_j}{\pi^{3/2} \sqrt{\log(n - 1)}} \log^2\left( \frac{4n}{n - j + 1} \right)\\
    & + 3\sqrt{\log (n) }\log\left(\frac{n}{n - j + 1}\right) + o(\log (n)).
\end{align*}

For the upper bound, we consider three asymptotic regimes,  
\begin{enumerate}
    \item \textbf{Case 1:}  $j = o(n)$.
    \item \textbf{Case 2:} $n - j + 1 = o(n)$; and
    \item \textbf{Case 3:} $j=\Theta(n)$ and $n-j + 1 = \Theta(n)$;
\end{enumerate}

The first two cases are easy, as in both cases, either $d_j$ or $d_{n-j+1}$ dominates and are $O(\log^{1/2}n)$. The tricky case is the third case. We first deal with the easy case before discussing the harder case:

\subsection*{Case 1:} When 
 $j = o(n)$, we use general upper bounds for $d_j$ and $d_{n - j + 1}$ using eq. \eqref{eq:d_j_definition}: 
\begin{equation*}
    d_j \le \sqrt{\alpha_{\infty} + \frac{\log (n - j + 1)}{\pi}} = O(\log^{1/2} n), \qquad d_{n - j + 1} \le \sqrt{\alpha_{\infty} + \frac{\log (j)}{\pi}} = o(\log^{1/2} n).
\end{equation*}

\subsection*{Case 2:}
  $n - j + 1 = o(n)$, we have:
\begin{equation*}
    d_j \le \sqrt{\alpha_{\infty} + \frac{\log (n - j + 1)}{\pi}} = o(\log^{1/2} n), \qquad d_{n - j + 1} \le \sqrt{\alpha_{\infty} + \frac{\log (j)}{\pi}} = O(\log^{1/2} n).
\end{equation*}

Substituting these bounds into the inequality, we obtain:
\begin{align*}
    \sum\limits_{k = 1}^{j} (\widetilde{B}_{j, k})^2 &\le o(\log^2 (n)) + \frac{3}{4\pi^2} \log^2 (n) + o(\log^2 (n)) + o(\log^{3/2} n) \\
    &\le \frac{3}{4\pi^2} \log^2(n) + o(\log^2 (n)).
\end{align*}

\subsection*{Case 3:} When $j = \Theta(n)$ and $n - j + 1 = \Theta(n)$, we can compute $d_j^2$ and $d_{n - j + 1}^2$ more precisely:
\begin{align}
\begin{split}
\label{eq:bound_on_d_j}    
    &d_j^2 = \alpha_{n - j + 1} + \frac{\log(n - j + 1)}{\pi} = \alpha_{\infty} + \frac{\log(n - j + 1)}{\pi} + o(1), \\
    &d_{n - j + 1}^2 = \alpha_j + \frac{\log(j)}{\pi} = \alpha_{\infty} + \frac{\log(j)}{\pi} + o(1).
\end{split}
\end{align}

Now, recall from  Lemma~\ref{lem:b_j_row_norm_bound}, we have  
\begin{align*}
    \sum\limits_{k = 1}^{j} (\widetilde{B}_{j, k})^2 & \le d_j^2d_{n - j + 1}^2 +  \frac{3}{4\pi^2} \log^2 \left(\frac{n - 1}{n - j + 1}\right) + \frac{\sqrt{6}d_j}{\pi^{3/2}\sqrt{\log (n - 1})}\log^2\left( \frac{4n}{n - j + 1} \right)\\
    &\qquad + 3\sqrt{\log (n)}\log\left(\frac{n}{n - j + 1}\right) + o(\log (n)) 
\end{align*}
Therefore, plugging the bounds in eq. \eqref{eq:bound_on_d_j} of $d_j$ and $d_{n-j-1}$, we get 
\begin{align*}
    \sum\limits_{k = 1}^{j} (\widetilde{B}_{j, k})^2 & \le \paren{\alpha_{\infty} + \frac{\log(n - j + 1)}{\pi} + o(1)}\paren{ \alpha_{\infty} + \frac{\log(j)}{\pi} + o(1)} +  \frac{3}{4\pi^2} \log^2 \left(\frac{n - 1}{n - j + 1}\right) \\
    & \qquad  + \frac{\sqrt{6}d_j}{\pi^{3/2}\sqrt{\log (n - 1})}\log^2\left( \frac{4n}{n - j + 1} \right) + 3\sqrt{\log (n)}\log\left(\frac{n}{n - j + 1}\right) + o(\log (n)) 
\end{align*}
Now, recall that $a_\infty \leq 1.0663$, and $\log(j)\leq \log(n)$. Furthermore, due to the case we are in, 
\[
\log \left(\frac{n - 1}{n - j + 1}\right) = \Theta(1) \quad \text{and} \quad \frac{\sqrt{6}d_j}{\sqrt{\pi^3\log (n - 1})}\log^2\left( \frac{4n}{n - j + 1} \right) = \Theta\paren{\frac{d_j^2}{\sqrt{\log(n)}}} = \Theta(\sqrt{\log(n)})
\]

Using this fact and expanding the product, we get 
\begin{align*}
    \sum\limits_{k = 1}^{j} (\widetilde{B}_{j, k})^2 
    &\le \frac{\log(n - j + 1) \log(j)}{\pi^2} + \frac{\alpha_{\infty}\left( \log(n - j + 1) + \log(j) \right)}{\pi}  + \Theta(1) +  \frac{3}{4\pi^2} \log^2 \left(\frac{n - 1}{n - j + 1}\right) \\
    & \qquad  + \frac{\sqrt{6}d_j}{\pi^{3/2}\sqrt{\log (n - 1})}\log^2\left( \frac{4n}{n - j + 1} \right) + 3\sqrt{\log (n)}\log\left(\frac{n}{n - j + 1}\right) + o(\log (n)) \\
    &\le \frac{\log(n - j + 1) \log(j)}{\pi^2} + \frac{\alpha_{\infty}\left( \log(n - j + 1) + \log(j) \right)}{\pi} + \underbrace{\Theta(1)  + \Theta(\log^{1/2} n)}_{=o(\log(n))} + o(\log (n)) 
\end{align*}
Now $\log(n-j+1),\log(j) \leq \log(n)$, so $\log(n-j+1)+\log(j) \leq 2\log(n)$. Further, $\log(n - j + 1) \log(j) \leq \log^2(n/2)$. This is because the maximum of the two is achieved when $n-j+1=j$ or when $j\in \{ \lfloor(n-1)/2\rfloor, \lceil (n-1)/2 \rceil \}$. 
That is, 
\begin{align*}
    \sum\limits_{k = 1}^{j} (\widetilde{B}_{j, k})^2 
    &\le \frac{\log(n - j + 1) \log(j)}{\pi^2} + \frac{2\alpha_{\infty}}{\pi} \log (n) + o(\log (n)) \\
    &\le \frac{\log^2(n/2)}{\pi^2} + \frac{2\alpha_{\infty}}{\pi} \log (n) + o(\log (n)) \\
    &\le \frac{\log^2 (n)}{\pi^2} + \frac{2(\alpha_{\infty} - \log(2)/\pi)}{\pi} \log (n) + o(\log (n)).
\end{align*}

In the last regime, the dominant term is $\frac{\log^2 (n)}{\pi^2}$. Since in the last regime the leading term is only $\frac{3\log^2 (n)}{4\pi^2}$ and in the second regime it is asymptotically smaller, the maximum over all $j$ is achieved in the last regime, concluding the upper bound.

\subsection*{Lower Bound}

For the lower bound, we simply consider $j = n/2$, which gives us the same expression, concluding the proof of Lemma~\ref{lem:rowNormOfB}. 
\end{proof}

\begin{corollary}
\label{cor:MaxSE_NSR}
The maximum squared error (MaxSE) of the Normalized Square Root (NSR) factorization is given by:
\begin{equation*}
    \mathrm{MaxSE}(\widetilde{B}, \widetilde{C}) = \alpha_{\infty} - \frac{\log (2)}{\pi} + \frac{\log (n)}{\pi} + o(1) \le \MaxSEUpperBoundExplicit.
\end{equation*}
\end{corollary}

\begin{proof}
The maximum column norm of the matrix $\widetilde{C}$ is $1$ by design, so the factorization error is equal to the maximum row norm of the matrix $\widetilde{B}$. Hence,
\begin{equation*}
    \sqrt{\max_{j} \sum\limits_{k = 1}^{j} (\widetilde{B}_{j, k})^2} 
    = \sqrt{\frac{\log^2 (n)}{\pi^2} + \frac{2(\alpha_{\infty} - \log(2)/\pi)}{\pi} \log (n) + o(\log (n))}= \frac{\log (n)}{\pi} + \left( \alpha_{\infty} - \frac{\log (2)}{\pi} \right) + o(1).
\end{equation*}

This proves the stated expression. Substituting the numerical value of $\alpha_{\infty}$ from Theorem~\ref{thm:square_root_bounds} completes the proof of Corollary~\ref{cor:MaxSE_NSR}.
\end{proof}
Before we compute the MeanSE of NSR, we prove the following auxiliary lemma:

\begin{lemma}
    \label{lem:sum_log_j_log_n_minus_j}
    The following asymptotic expansion holds:
    \begin{equation*}
        \frac{1}{n} \sum\limits_{j = 1}^{n} \log (j)\log(n + 1 - j) = \log^2(n) - 2\log (n) + 2 - \frac{\pi^2}{6} + o(1).
    \end{equation*}
\end{lemma}

\begin{proof}
We first rewrite the sum as follows:
\begin{align*}
    \frac{1}{n} \sum\limits_{j = 1}^{n} \log (j)\log(n + 1 - j) &= \frac{1}{n} \sum\limits_{j = 1}^{n} \left(\log (n) + \log \paren{\tfrac{j}{n}}\right)\left(\log (n) + \log \tfrac{n + 1 - j}{n}\right) \\
    &= \log^2 (n) + \frac{\log (n)}{n} \sum\limits_{j = 1}^{n} \left[\log \paren{\tfrac{j}{n}} + \log \tfrac{n + 1 - j}{n}\right] + \frac{1}{n} \sum\limits_{j = 1}^{n} \log \left(\tfrac{j}{n}\right) \log \left(\tfrac{n + 1 - j}{n}\right) \\
    &= \log^2 (n) + \frac{2\log (n)}{n} \sum\limits_{j = 1}^{n} \log \paren{\tfrac{j}{n}} + \frac{1}{n} \sum\limits_{j = 1}^{n} \log \left(\tfrac{j}{n}\right) \log \left(\tfrac{n + 1 - j}{n}\right).
\end{align*}

Now we recognize that the two sums are Riemann sums and approximate them by their corresponding integrals, up to $o(1)$ terms:

\begin{equation*}
    \frac{1}{n} \sum\limits_{j = 1}^{n} \log (j)\log(n + 1 - j) = \log^2 (n) + 2\log (n) \int\limits_{0}^{1} \log (x) \, \mathsf dx + \int\limits_{0}^{1} \log (x) \log (1 - x) \, \mathsf dx + o(1).
\end{equation*}

The first integral can be computed explicitly:
\begin{equation*}
    \int \log (x) \, \mathsf dx = x \log (x) - x + \mathsf{const.},
\end{equation*}
and evaluating from 0 to 1 gives $-1$. The second integral is more involved; we compute it using a method suggested in \cite{stackexchangeintegral}:

\begin{align*}
    \int\limits_0^1 \log(1-x) \log (x) \, \mathsf dx &= \int\limits_0^1 \left(-\sum_{n=1}^\infty \frac{x^n}{n}\right) \log (x) \, \mathsf dx 
    = \sum_{n=1}^\infty \frac{-1}{n} \int\limits_0^1 x^n \log (x) \, \mathsf dx \\
    &= \sum_{n=1}^\infty \frac{1}{n} \cdot \frac{1}{(n+1)^2} 
    = \sum_{n=1}^\infty \left[ \frac{1}{n} - \frac{1}{n+1} - \frac{1}{(n+1)^2} \right] = 2 - \sum_{n=1}^\infty \frac{1}{n^2} 
    = 2 - \frac{\pi^2}{6}.
\end{align*}
Substituting the integrals back into the equation, we get Lemma~\ref{lem:sum_log_j_log_n_minus_j}.
\end{proof}
\begin{theorem}
\label{thm:meanSE_NSR}
    The MeanSE of the Normalized Square Root (NSR) factorization is given by the following identity:
    \begin{equation*}
        \mathrm{MeanSE}(\widetilde{B}, \widetilde{C}) = \MeanSEUpperBound,
    \end{equation*}
where  $\alpha_{\infty}$ is defined in Theorem~\ref{thm:square_root_bounds} and equal to $\alpha_{\infty} = \frac{\gamma + \log 16}{\pi}$, see Corollary~\ref{cor:alpha_infinity}. 
\end{theorem}

\begin{proof}
    The squared MeanSE error is given by:
    \begin{equation*}
        \text{MeanSE}^2(\widetilde{B}, \widetilde{C}) = \frac{1}{n}\sum\limits_{j = 1}^{n}\sum\limits_{k = 1}^{j} \widetilde{B}^2_{j, k}
    \end{equation*}

    We bounded the squared $j$-th row norm of the matrix $\widetilde{B}$ in Lemma~\ref{lem:b_j_row_norm_bound}:
    \begin{align*}
        d_j^2 d_{n - j + 1}^2 \le \sum\limits_{k = 1}^{j} (\widetilde{B}_{j, k})^2 &\le d_j^2 d_{n - j + 1}^2 + \frac{3}{4\pi^2} \log^2 \left(\frac{n - 1}{n - j + 1}\right) + \frac{\sqrt{6}d_j}{\pi^{3/2} \sqrt{\log(n - 1)}} \log^2\left( \frac{4n}{n - j + 1} \right)\\
        &\qquad + 5\sqrt{\log (n)} \log\left(\frac{n}{n - j + 1}\right) + o(\log (n))
    \end{align*}

    Now we need to average these values over the index $j$. First, analyze the main asymptotic term:
    \begin{align*}
        \sum\limits_{j = 1}^n d_j^2 d_{n - j + 1}^2 &= \sum\limits_{j = 1}^n \left( \frac{\log(n - j + 1)}{\pi} + \alpha_{n - j + 1} \right) \left( \frac{\log(j)}{\pi} + \alpha_j \right)\\
        &= \frac{1}{\pi^2} \sum\limits_{j = 1}^{n} \log(n - j + 1) \log(j) + \frac{1}{\pi} \sum\limits_{j = 1}^{n} \left[ \alpha_j \log(n - j + 1) + \alpha_{n - j + 1} \log(j) \right] + \sum\limits_{j = 1}^{n} \alpha_j \alpha_{n - j + 1}.
    \end{align*}

    The first sum we analyzed in Lemma~\ref{lem:sum_log_j_log_n_minus_j}. For the second sum, we use the decomposition $\alpha_j = \alpha_{\infty} - \epsilon_j$, where $0 < \epsilon_j \le \frac{1}{2j}$, as given by Theorem~\ref{thm:square_root_bounds}. The third sum is then trivially $\alpha_{\infty}^2 + o(1)$. We thus rewrite the second sum as:
    \begin{align*}
        &\frac{2\alpha_{\infty}}{n\pi} \sum\limits_{j = 1}^{n} \log(j) - \frac{1}{\pi n} \sum\limits_{j = 1}^{n} \left[ \epsilon_j \log(n - j + 1) + \epsilon_{n - j + 1} \log(j) \right] = \frac{2\alpha_{\infty}}{n\pi} \log(n!) - \frac{2}{\pi n} \sum\limits_{j = 1}^{n} \epsilon_j \log(n - j + 1)\\
        &\quad = \frac{2\alpha_{\infty} (\log (n) - 1)}{\pi} + o(1) - \frac{2}{\pi n} \sum\limits_{j = 1}^{n} \epsilon_j \log(n - j + 1).
    \end{align*}
We now show that the final term is $o(1)$:
\begin{align*}
    0 < \frac{2}{\pi n} \sum\limits_{j = 1}^{n} \epsilon_j \log(n - j + 1) &\le \frac{1}{\pi n} \sum\limits_{j = 1}^{n} \frac{\log(n - j + 1)}{j} = \frac{H_n \log (n)}{\pi n} + \frac{1}{\pi n^2} \sum\limits_{j = 1}^{n} \frac{\log(1 - \tfrac{j - 1}{n})}{j / n}\\
    &\le o(1) + \frac{1}{\pi n} \int\limits_0^1 \frac{\log(1 - x)}{x}\; \mathsf dx + o(n^{-1})\\
    &= o(1) - \frac{\mathrm{Li}_2(1)}{\pi n} = o(1) - \frac{\zeta(2)}{\pi n} = o(1),
\end{align*}
where $\mathrm{Li}_2(x)$ is the dilogarithmic function. 

Combining all terms, we get:
\begin{align*}
    \frac{1}{n} \sum\limits_{j = 1}^n d_j^2 d_{n - j + 1}^2 &= \frac{\log^2 (n) - 2\log (n) + 2 }{\pi^2} - \frac{1}{6} + \frac{2\alpha_{\infty}(\log (n) - 1)}{\pi} + \alpha_{\infty}^2 + o(1)\\
    &= \frac{\log^2 (n)}{\pi^2} + \frac{2}{\pi} \left( \alpha_{\infty} - \frac{1}{\pi} \right) \log (n) + o(\log (n)).
\end{align*}

The only remaining part is to prove that the contribution from the average of the extra upper bound terms is $o(\log (n))$:
\begin{align*}
    &\frac{1}{n} \sum\limits_{j = 1}^{n} \left[ \frac{3}{4\pi^2} \log^2\left(\frac{n - 1}{n - j + 1}\right) + \frac{\sqrt{6} d_j}{\pi^{3/2} \sqrt{\log(n - 1)}} \log^2\left( \frac{4n}{n - j + 1} \right) +  3\sqrt{\log (n)} \log\left(\frac{n}{n - j + 1}\right) + o(\log (n)) \right]\\
    &\qquad \le \left( \frac{3}{4\pi^2} + \frac{\sqrt{6}}{\pi^2} + o(1) \right) \cdot \frac{1}{n} \sum\limits_{j = 1}^{n} \log^2\left( \frac{4n}{n - j + 1} \right) + \frac{3 \sqrt{\log (n)}}{n} \sum\limits_{j = 1}^{n} \log\left( \frac{n}{n - j + 1} \right) + o(\log (n))\\
    &\qquad = \left( \frac{3}{4\pi^2} + \frac{\sqrt{6}}{\pi^2} + o(1) \right) \cdot \frac{1}{n} \sum\limits_{j = 1}^{n} \left( \log(4n) - \log(n - j + 1) \right)^2 + 3 \sqrt{\log (n)} \left[ \log (n) - \frac{\log (n)!}{n} \right] + o(\log (n))\\
    &\qquad \le \left( \frac{1}{3} + o(1) \right) \cdot \frac{1}{n} \sum\limits_{j = 1}^{n} \left[ \log^2(4n) - 2\log(4n) \log(n - j + 1) + \log^2(n - j + 1) \right] + o(\log (n))\\
    &\qquad = \left( \frac{1}{3} + o(1) \right) \left[ \log^2(4n) - 2\log(4n) \cdot \frac{\log (n)!}{n} + \frac{1}{n} \sum\limits_{j = 1}^{n} \log^2 (j) \right] + o(\log (n))\\
    &\qquad = \left( \frac{1}{3} + o(1) \right) \left[ \log^2 (n) + 2 \log(4) \log (n) - 2\log(4n)(\log (n) - 1) + \frac{1}{n} \sum\limits_{j = 1}^{n} \log^2 (j) \right] + o(\log (n))
\end{align*}

We use an integral bound for the sum of $\log^2 (j)$:
\begin{align*}
    \frac{1}{n} \sum\limits_{j = 1}^{n} \log^2 (j) &< \frac{1}{n} \int\limits_{1}^{n + 1} \log^2 (x) \; \mathsf dx = \frac{(n + 1)(\log^2(n + 1) - 2\log(n + 1) + 2) - 2}{n}\\
    &= \log^2 (n) - 2\log (n) + o(\log (n))
\end{align*}

This gives us 
\begin{align*}
    &\left( \frac{1}{3} + o(1) \right) \left[ \log^2 (n) + 2 \log 4 \cdot \log (n) - 2 \log(4n)(\log (n) - 1) + \log^2 (n) - 2\log (n) \right] + o(\log (n))\\
    &\qquad = \left( \frac{1}{3} + o(1) \right) \cdot 2\log (2) + o(\log (n)) = o(\log (n)),
\end{align*}
concluding the proof of Theorem~\ref{thm:meanSE_NSR}.
\end{proof}

\section{Improved Analysis of Known Factorizations}

\subsection{Group Algebra Factorization}
\label{sub:group_algebra}
Let $\omega=e^{\iota \pi/n}$ be a  $2n$-th root of unity.  A recent work by Henzinger and Upadhyay \cite{henzinger2025improved} presented the following factorization of the matrix $\counting = L R$ into two complex matrices, where 
\begin{align}
\label{eq:group_algebra_factorization}
    L = L_{\mathsf{HU}} :=  \begin{pmatrix}
    b_f(\omega^0) &  \dots & b_f(\omega^{2n - 1})\\
    \vdots  & \ddots & \vdots\\
    b_f(\omega^{-n + 1})  & \dots & b_f(\omega^{n})\\
\end{pmatrix} \in \mathbb{C}^{n \times 2n}, \quad 
R = R_{\mathsf{HU}} :=\begin{pmatrix}
    b_f(\omega^0)  & \dots & b_f(\omega^{n - 1})\\
    \vdots  & \ddots & \vdots\\
    b_f(\omega^{-2n + 1})  & \dots & b_f(\omega^{-n}) \\
\end{pmatrix} \in \mathbb{C}^{2n \times n},
\end{align}
with the values $b_f(\omega^k)$ given by the function: 
\begin{equation*}
        b_f(x) = \frac{1}{2n} \sum\limits_{l = 0}^{2n - 1} x^l \left(\sum\limits_{k = 0}^{n - 1} \omega^{kl} \right)^{1/2}.
    \end{equation*}

We note that in Section~\ref{sec:introduction}, we refer to these matrices as $L_{\mathsf{HU}}$ and $R_{\mathsf{HU}}$, respectively. We prefer this notation in this section so as not to use double subscripts when denoting the entries of these matrices. 
    
The values $b_f(\omega^{k})$ are $2n$-periodic, i.e., $b_f(\omega^{k}) = b_f(\omega^{k+2n})$ for all $k \in \mathbb Z$. This implies that the matrices $L$ and $R$ row- and column-\textbf{circulant}, respectively.

\begin{lemma}
\label{lem:dft_decomposition}
    The Group Algebra factorization $L, R$ (see equation~\ref{eq:group_algebra_factorization}) is equivalent to matrices $\widetilde{L} \in \mathbb{C}^{n \times 2n}$ and $\widetilde{R} \in \mathbb{C}^{2n \times n}$ given by:
    \begin{equation*}
        \widetilde{L}_{j, k} = \begin{cases}
            \frac{1}{\sqrt{2}}, & k = 0, \\
            \frac{\omega^{jk}}{\sqrt{n(1 - \omega^{-k})}}, & k \equiv 1 \pmod 2, \\
            0, & k \equiv 0 \pmod 2,\; k \ne 0,
        \end{cases}
        \hspace{1cm}
        \widetilde{R}_{j, k} = \begin{cases}
            \frac{1}{\sqrt{2}}, & j = 0, \\
            \frac{\omega^{-jk}}{\sqrt{n(1 - \omega^{-j})}}, & j \equiv 1 \pmod 2, \\
            0, & j \equiv 0 \pmod 2,\; j \ne 0,
        \end{cases}
    \end{equation*}
    where $\omega = e^{i\pi/n}$ is a $2n$-th root of unity.
\end{lemma}
\begin{proof}
    Consider the Discrete Fourier Transform (DFT) matrix $F_{2n}$ defined as:
    \begin{equation*}
        F_{2n} = \frac{1}{\sqrt{2n}} \begin{pmatrix}
            1 & 1 & \dots & 1 \\
            1 & \omega^{-1} & \dots & \omega^{-(2n - 1)} \\
            \vdots & \vdots & \ddots & \vdots \\
            1 & \omega^{-(2n - 1)} & \dots & \omega^{-(2n - 1)(2n - 1)}
        \end{pmatrix},
    \end{equation*}
    where $\omega = e^{i\pi / n}$ is a $2n$-th root of unity. This matrix is unitary, and therefore does not change the Frobenius norm of $L$ nor the maximum column norm of $R$, preserving the factorization error. Thus, we define a new factorization as $\widetilde{L} = L  \times F_{2n}^{*}$ and $\widetilde{R} = F_{2n} \times R$, where $F^{*}_{2n}$ denotes the conjugate transpose of $F_{2n}$.

    To compute the product, we begin with $\widetilde{L}$:
    \begin{align*}
        \widetilde{L}_{j, k} &= \sum_{t = 0}^{2n - 1} L_{j, t} (F_{2n}^*)_{t, k}
        = \frac{1}{\sqrt{2n}} \sum_{t = 0}^{2n - 1} b_f(\omega^{-j + t}) \omega^{tk} = \frac{1}{(2n)^{3/2}} \sum_{t = 0}^{2n - 1} \omega^{tk} \left[\sum_{l = 0}^{2n - 1} \omega^{(t - j)l} \left(\sum_{m = 0}^{n - 1} \omega^{ml}\right)^{1/2} \right] \\
        &= \frac{1}{(2n)^{3/2}} \sum_{t = 0}^{2n - 1} \omega^{tk} \left[ \sqrt{n} + \sqrt{2} \sum_{l = 0}^{n - 1} \frac{\omega^{(t - j)(2l + 1)}}{(1 - \omega^{2l + 1})^{1/2}} \right]\\
        &= \frac{\mathbbm{1}_{k = 0}}{\sqrt{2}} + \frac{1}{2n^{3/2}} \sum_{l = 0}^{n - 1} \frac{\omega^{-j(2l + 1)}}{(1 - \omega^{2l + 1})^{1/2}} \sum_{t = 0}^{2n - 1} \omega^{t(2l + 1 + k)} \\
        &= \frac{\mathbbm{1}_{k = 0}}{\sqrt{2}} + \frac{1}{\sqrt{n}} \sum_{l = 0}^{n - 1} \frac{\omega^{-j(2l + 1)}}{(1 - \omega^{2l + 1})^{1/2}} \mathbbm{1}_{2l + 1 = 2n - k} = \frac{\mathbbm{1}_{k = 0}}{\sqrt{2}} + \frac{\omega^{jk}}{(1 - \omega^{-k})^{1/2}} \mathbbm{1}_{\text{$k$ odd}}.
    \end{align*}

    Similarly, for matrix $\widetilde{R}$:
    \begin{align*}
        \widetilde{R}_{j, k} &= \sum_{t = 0}^{2n - 1} (F_{2n})_{j, t} R_{t, k}
        = \frac{1}{\sqrt{2n}} \sum_{t = 0}^{2n - 1} \omega^{-jt} b_f(\omega^{-t + k}) = \frac{1}{(2n)^{3/2}} \sum_{t = 0}^{2n - 1} \omega^{-jt} \left[\sum_{l = 0}^{2n - 1} \omega^{(k - t)l} \left(\sum_{m = 0}^{n - 1} \omega^{ml}\right)^{1/2} \right] \\
        &= \frac{1}{(2n)^{3/2}} \sum_{t = 0}^{2n - 1} \omega^{-jt} \left[ \sqrt{n} + \sqrt{2} \sum_{l = 0}^{n - 1} \frac{\omega^{(k - t)(2l + 1)}}{(1 - \omega^{2l + 1})^{1/2}} \right] \\
        &= \frac{\mathbbm{1}_{j = 0}}{\sqrt{2}} + \frac{1}{2n^{3/2}} \sum_{l = 0}^{n - 1} \frac{\omega^{k(2l + 1)}}{(1 - \omega^{2l + 1})^{1/2}} \sum_{t = 0}^{2n - 1} \omega^{-t(j + 2l + 1)} \\
        &= \frac{\mathbbm{1}_{j = 0}}{\sqrt{2}} + \frac{1}{\sqrt{n}} \sum_{l = 0}^{n - 1} \frac{\omega^{k(2l + 1)}}{(1 - \omega^{2l + 1})^{1/2}} \mathbbm{1}_{2l + 1 = 2n - j} = \frac{\mathbbm{1}_{j = 0}}{\sqrt{2}} + \frac{\omega^{-jk}}{(1 - \omega^{-j})^{1/2}} \mathbbm{1}_{\text{$j$ odd}}.
    \end{align*}
This concludes the proof of Lemma~\ref{lem:dft_decomposition}.
\end{proof}

\begin{corollary}
\label{cor:L_R_decomposition}
The matrices $L$ and $R$ defined in eq.~\eqref{eq:group_algebra_factorization} can be expressed as follows:
\begin{equation*}
    L = P F_{2n}^{*} \Sigma^{1/2} F_{2n}, \quad R = F_{2n}^{*} \Sigma^{1/2} F_{2n} P^{\top},
\end{equation*}
where $\Sigma \in \mathbb{C}^{2n \times 2n}$ is a diagonal matrix, and $P \in \mathbb{R}^{2n \times n}$ is the projection matrix onto the first $n$ coordinates, given by:
\begin{equation*}
    \Sigma =  \text{diag}\left(n,\; \frac{2}{1 - \omega^{-1}}, \ldots,  \frac{1 - \omega^{-kn}}{1 - \omega^{-k}}, \ldots, \frac{2}{1 - \omega^{-(2n - 1)}},\;  0\right), \quad 
    P = \begin{pmatrix}
        I_n & 
        \mathbf{0}^{n\times n}
    \end{pmatrix},
\end{equation*}
where $\omega = e^{i\pi/n}$ is a $2n$-th root of unity.
\end{corollary}

\begin{proof}
From Lemma~\ref{lem:dft_decomposition}, we have that $L = \widetilde{L} F_{2n}$ and $R = F_{2n}^{*} \widetilde{R}$, where
\begin{equation*}
    \widetilde{L} = 
    \begin{pmatrix}
        1 & 1 & \cdots & 1\\
        1 & \omega & \dots & \omega^{2n - 1}\\
        \vdots & \vdots & \ddots & \vdots\\
        1 & \omega^{n - 1} & \cdots & \omega^{(n - 1)(2n - 1)} 
    \end{pmatrix} \times
    \begin{pmatrix}
        \frac{1}{\sqrt{2}} & 0 & \dots & 0\\
        0 & \frac{1}{\sqrt{n}(1 - \omega^{-1})} & \dots & 0\\
        \vdots & \vdots & \ddots & \vdots\\
        0 & 0 & \dots & 0
    \end{pmatrix}.
\end{equation*}
The first matrix is a truncated $F_{2n}^{*}$ scaled by $\sqrt{2n}$. When the second matrix is multiplied by $\sqrt{2n}$, we obtain exactly $\sqrt{\Sigma}$; hence, $\widetilde{L} = P F_{2n}^{*} \Sigma^{1/2}$. Analogously, we get $\widetilde{R} = \Sigma^{1/2} F_{2n} P^{\top}$, which completes the proof of Corollary~\ref{cor:L_R_decomposition}.
\end{proof}

We note that this form of the matrices $L$ and $R$ corresponds to the factorizations $\mathbf{B_F}$ and $\mathbf{C_F}$ in Section K of \cite{choquette2022multi}.

\begin{lemma}
The Group Algebra factorization defined by factors in eq.~\eqref{eq:group_algebra_factorization} is equivalent to
\begin{equation*}
    L = P M_{\text{circ}}^{1/2}, \quad R = M_{\text{circ}}^{1/2} P^{\top},
\end{equation*}
where $M_{\text{circ}} \in \mathbb{R}^{2n \times 2n}$ is a circulant extension of the matrix $M$, and $P$ is the projection matrix onto the first $n$ coordinates. These matrices are given by:
\begin{equation*}
    M_{\text{circ}} = \begin{pmatrix}
        1 & 0 & \dots & 0 & 0 & 1 & \dots & 1\\
        1 & 1 & \dots & 0 & 0 & 0 & \dots & 1\\
        \vdots & \vdots & \ddots & \vdots & \vdots & \vdots & \ddots & \vdots\\
        1 & 1 & \dots & 1 & 0 & 0 & \dots & 0\\
        0 & 1 & \dots & 1 & 1 & 0 & \dots & 0\\
        \vdots & \vdots & \ddots & \vdots & \vdots & \vdots & \ddots & \vdots\\
        0 & 0 & \dots & 0 & 1 & 1 & \dots & 1
    \end{pmatrix}, 
    \quad 
    P = \begin{pmatrix}
        I_n & 
        \mathbf{0}^{n\times n}
    \end{pmatrix}.
\end{equation*}
\end{lemma}

\begin{proof}
First, we show that $M_{\text{circ}} = F_{2n}^{*} \Sigma F_{2n}$, with $\Sigma$ defined in Corollary~\ref{cor:L_R_decomposition}; for this, we refer to Theorem J.1 and Equation (44) of \cite{choquette2022multi}. Then, we show that $M_{\text{circ}}^{1/2} = F_{2n}^{*} \Sigma^{1/2} F_{2n}$ by squaring both sides. Multiplying both sides on the left by $P$ proves that $L = P M_{\text{circ}}^{1/2}$. The proof for matrix $R$ is analogous. We thus have the lemma.
\end{proof}
\begin{lemma}
\label{lem:group_algebra_max_se}
    The MaxSE and MeanSE errors of the Group Algebra Factorization defined by factors in eq.~\eqref{eq:group_algebra_factorization} are given by the following identity:
    \begin{equation*}
        \mathrm{MaxSE}(L, R) = \mathrm{MeanSE}(L, R) = \frac{\log (n)}{\pi} + \underbrace{\frac{1}{\pi}\paren{\log \paren{\tfrac{8}{\pi}} + \gamma} + \frac{1}{2}}_{\approx 0.981} + o(1),
    \end{equation*}
    where $\gamma \approx 0.57721$ is the Euler–Mascheroni constant.  
\end{lemma}

\begin{proof}
    In Lemma 7 of \cite{henzinger2025binned}, it was shown that:
    \begin{equation*}
        \|L\|_{2\to \infty}^2 = \|R\|_{1 \to 2}^2 = \frac{1}{2} + \frac{1}{2n} \sum\limits_{l = 1}^{n} \frac{1}{\sin\left(\tfrac{\pi(2l - 1)}{2n}\right)},
    \end{equation*}
    and since the matrix $L$ is circulant, all its rows have the same norm. Therefore, we obtain:
    \begin{equation*}
         \text{MaxSE}(L, R) = \text{MeanSE}(L, R) = \frac{1}{2} + \frac{1}{2n} \sum\limits_{l = 1}^{n} \frac{1}{\sin\left(\tfrac{\pi(2l - 1)}{2n}\right)}.
    \end{equation*}
    Now we analyze the sum. We rewrite it as follows:
    \begin{equation*}
        \frac{1}{2n} \sum\limits_{l = 1}^{n} \frac{1}{\sin\left(\tfrac{\pi(2l - 1)}{2n}\right)} 
        = \frac{1}{2n} \sum\limits_{l = 1}^{2n - 1} \frac{1}{\sin\left(\tfrac{\pi \ell}{2n}\right)} 
        - \frac{1}{2n} \sum\limits_{l = 1}^{n - 1} \frac{1}{\sin\left(\tfrac{\pi 2l}{2n}\right)} 
        = G(2n) - \tfrac{1}{2}G(n),
    \end{equation*}

From Lemma~\ref{lem:boundonG(n)}, we get 
\begin{equation*}
     G(2n) - \tfrac{1}{2} G(n) = \frac{\log (n)}{\pi} + \frac{\gamma + \log \paren{\tfrac{8}{\pi}}}{\pi} + o(1),
\end{equation*}

Equating the two concludes the proof of Lemma~\ref{lem:group_algebra_max_se}.
\end{proof}

\subsection{Square Root Factorization}
\label{sec:square_root_factorization}

The \emph{Square Root Factorization} of the matrix $\counting$ is defined as the matrix square root $C = \counting^{1/2}$, whose entries were shown in \cite{henzinger2022constant} to be
\begin{equation}
    C =
    \begin{pmatrix}
        r_0 & 0 & \dots & 0 \\
        r_1 & r_0 & \dots & 0 \\
        \vdots & \vdots & \ddots & \vdots \\
        r_{n - 1} & r_{n - 2} & \dots & r_0
    \end{pmatrix},
\end{equation}
where $r_0 = 1$ and, for $k \ge 1$, $r_k = \left|\binom{-1/2}{k}\right| = \frac{1}{4^k}\binom{2k}{k}$. The coefficients $r_k$ satisfy the classical bounds
\begin{equation}
    \frac{1}{\sqrt{\pi (k + 1)}} \le r_k \le \frac{1}{\sqrt{\pi k}},
\end{equation}
which follow directly from Wallis' inequality.  

The MaxSE error of the square root factorization can be bounded by
\begin{equation}
    \mathrm{MaxSE}(C, C) \le \frac{\log (n)}{\pi} + \underbrace{1 + \frac{\gamma}{\pi}}_{\approx 1.18},
\end{equation}
where $\gamma$ denotes the Euler–Mascheroni constant~\cite[Theorem 2.7]{pillutla2025correlated}. In the next theorem, we improve upon this bound.

\begin{theorem}
\label{thm:square_root_bounds}
The MaxSE error of the square root factorization of size $n$ is given by
\begin{equation*}
    \sum\limits_{j = 0}^{n - 1} r_j^2 = \alpha_n + \frac{\log (n)}{\pi} = \alpha_{\infty} - \epsilon_n + \frac{\log (n)}{\pi},
\end{equation*}
where $0 \le \epsilon_n \le \frac{1}{5n}$, and $\alpha_{\infty} \approx 1.0663$. In particular, $1 \leq a_n \leq 1.0663$.
\end{theorem}

\begin{proof}
We begin by upper bounding $\alpha_n$ as 
\begin{align*}
    \alpha_n = \sum\limits_{j = 0}^{n - 1} r_j^2 - \frac{\log (n)}{\pi} \le 1 + \frac{1}{\pi} \sum\limits_{j = 1}^{n - 1} \frac{1}{j} - \frac{\log (n)}{\pi} 
    \le 1 + \frac{1}{\pi}.
\end{align*}

From Lemma~\ref{lem:monotonicityofa_n}, the  sequence $a_n$ is monotonically increasing and bounded from above, it converges. The limiting constant $\alpha_{\infty}$ can be computed numerically to any desired precision. The lower bound $\epsilon_n \ge 0$ follows directly from monotonicity. 

To prove the upper bound $\epsilon_n = \alpha_{\infty} - \alpha_n \le \frac{1}{5n}$, we rewrite $\alpha_n$ as follows:
\begin{equation*}
    \alpha_n = \sum\limits_{j = 0}^{n - 1} r_j^2 - \frac{\log (n)}{\pi} 
    = \sum\limits_{j = 0}^{n - 1} \left[ r_j^2 - \frac{1}{\pi(j+1)} \right] 
    + \frac{H_n}{\pi} - \frac{\log (n)}{\pi} 
    = \sum\limits_{j = 0}^{n - 1} \left[ r_j^2 - \frac{1}{\pi(j+1)} \right] 
    + \frac{\gamma}{\pi} + \frac{1}{2\pi n} - \frac{s_n}{\pi},
\end{equation*}
where in the last step we used the harmonic number expansion $H_n = \log (n) + \gamma + \frac{1}{2n} - s_n$, with $0 \le s_n \le \frac{1}{8n^2}$.

Hence, the difference between $\alpha_{\infty}$ and $\alpha_n$ is
\begin{align*}
    \epsilon_n = \alpha_{\infty} - \alpha_n 
    &= \sum\limits_{j = n}^{\infty} \left[ r_j^2 - \frac{1}{\pi(j+1)} \right] - \frac{1}{2\pi n} + \frac{s_n}{\pi} \\
    &\le \frac{1}{\pi} \sum\limits_{j = n}^{\infty} \left( \frac{1}{j} - \frac{1}{j+1} \right) - \frac{1}{2\pi n} + \frac{1}{8\pi n^2} = \frac{1}{2\pi n} + \frac{1}{8\pi n^2} < \frac{1}{5n},
\end{align*}
which concludes the proof of Theorem~\ref{thm:square_root_bounds}.
\end{proof}

The only previously known bound for the MeanSE error of the square root factorization was obtained by upper-bounding it via the MaxSE error. In the following theorem, we show that tighter constants can be achieved by computing the MeanSE exactly.

\begin{theorem}
    The MeanSE error of the square root factorization is given by:
    \begin{equation*}
        \mathrm{MeanSE}(C,C) = \frac{\log (n)}{\pi} + \underbrace{\alpha_{\infty} - \frac{1}{2\pi}}_{\approx 0.907} + o(1).
    \end{equation*}
\end{theorem}

\begin{proof}
    The squared normalized Frobenius norm $\tfrac{1}{n}\|C\|_F^2$ is given by
    \begin{equation*}
        \frac{1}{n} \sum\limits_{k = 1}^{n} \sum\limits_{j = 0}^{k - 1} r_j^2 = \frac{1}{n} \sum\limits_{k = 1}^{n} \left[\frac{\log k}{\pi} + \alpha_k\right] = \frac{\log (n)!}{\pi n} + \frac{1}{n} \sum\limits_{k = 1}^{n}\alpha_k = \frac{\log (n) - 1 + o(1)}{\pi} + \alpha_{\infty} + o(1), 
    \end{equation*}
where in the last identity we used that the sequence $\alpha_k$ converges to $\alpha_{\infty}$ (therefore their averages also converge to $\alpha_{\infty}$) and Stirling's approximation for $\log (n!)$. The maximum column norm of matrix $C$ is given by:
\begin{equation*}
    \|C\|_{1\to 2}^2 = \sum\limits_{j = 0}^{n - 1} r_j^2 = \frac{\log (n)}{\pi} + \alpha_{\infty} + o(1).
\end{equation*}

Thus, the MeanSE is given by:

\begin{equation*}
    \text{MeanSE}(C,C) = \sqrt{\frac{\log ^2n}{\pi^2} + \frac{2\log (n)}{\pi} \left(\alpha_{\infty} - \frac{1}{2\pi}\right) + o(\log (n))} = \frac{\log (n)}{\pi} + \underbrace{\alpha_{\infty} - \frac{1}{2\pi}}_{\approx 0.907} + o(1).
\end{equation*}
This concludes the proof of theorem.
\end{proof}

\begin{lemma}
The square root factorization $C \times C = \counting$ is equivalent to 
\begin{equation*}
F_n S^{1/2} \times S^{1/2} F_n^{*} = \counting, \quad \text{where} \quad S = F_n^{*} \counting F_n.
\end{equation*}
\end{lemma}

\begin{proof}
We first show that $F_n^{*} C F_n = S^{1/2}$ by squaring both sides of the equation. Then, the new factorization is given by $C F_n \times F_n^{*} C$. Since multiplying a factorization by a unitary matrix $F_n$ and its inverse does not change the factorization error, the result follows.
\end{proof}

\subsection{Discussion on Landau  constants}
\label{sec:landauConstant}
In this work, we are the first in the matrix factorization community to observe that the sums of $r_k^2$ have been studied previously in the context of mathematical analysis, where they are known as the \emph{Landau constants}, defined by
\begin{equation}
    G_n = \sum\limits_{k = 0}^{n} \frac{1}{16^k}\binom{2k}{k}^2 =\sum_{k = 0}^{n} r_k^2.
\end{equation}
In his original work \cite{landau1916darstellung},  Landau showed that $G_n \sim \frac{\log (n)}{\pi}$.  The following more precise expansion was later obtained.

\begin{theorem}[Watson \cite{watson1930constants}]
The Landau constant has the following closed form expression:
    \begin{equation}
        G_n = \sum_{k = 0}^{n} r_k^2
        = \frac{\log (n + 1)}{\pi} + \frac{\gamma + \log 16}{\pi} - \frac{1}{4\pi (n + 1)} + \mathcal{O}\!\left(\frac{1}{n^2}\right).
    \end{equation}
\end{theorem}

\begin{corollary}
\label{cor:alpha_infinity}
Using $\sum_{k = 0}^{n - 1} r_k^2 = G_{n - 1}$, we deduce that
\[
\alpha_{\infty} = \frac{\gamma + \log 16}{\pi}.
\]
\end{corollary}

In the work of Brutman \cite{brijtman1982sharp}, the following bounds were established:
\begin{equation}
    1 + \frac{\log (n + 1)}{\pi} \le G_n \le \alpha_{\infty} + \frac{\log (n + 1)}{\pi}.
\end{equation}

In subsequent work, Falaleev \cite{falaleev1991inequalities} improved the lower bound:
\begin{equation}
    \alpha_{\infty} + \frac{\log (n + 3/4)}{\pi} \le G_n \le \alpha_{\infty} + \frac{\log (n + 1)}{\pi}.
\end{equation}

These results are sufficient to cover the content of Theorem~\ref{thm:square_root_bounds}. However, our proof was discovered independently and differs technically from the previous work. For the sake of completeness and self-containment, we choose to present our proof rather than relying only on the existing theory of Landau constants, whose topic is both rich and intriguing. For more modern treatments, including sharper inequalities and higher-order expansions, we refer the reader to \cite{chen2014asymptotic, chen2017unified,zhao2009some}.

\section*{Acknowledgements.} 

We thank Joel Andersson for valuable comments on the early version of the draft.

% In lexicographic order

\erclogowrapped{5\baselineskip}Monika Henzinger:  This project has received funding from the European Research Council (ERC) under the European Union's Horizon 2020 research and innovation programme (Grant agreement No.\ 101019564) and the Austrian Science Fund (FWF) grant DOI 10.55776/Z422.

\sloppy 
Nikita Kalinin: This work is
supported in part by the Austrian Science Fund (FWF) [10.55776/COE12]. A part of this work was done while visiting  University of Copenhagen.

Jalaj Upadhyay: This project was supported in my part by NSF CNS 2433628, Google Seed Fund grant, Google Research Scholar Award, Dean Research Seed Fund, and Decanal Research Grant. A part of work was done while visiting Institute of Science and Technology-Austria. 

\bibliographystyle{plain}
\bibliography{references}

\end{document}